\documentclass[reqno]{amsart}

\usepackage{graphicx}
\usepackage{amsfonts}
\usepackage{dsfont} 

\newtheorem{lemma}{Lemma}
\newtheorem{theorem}{Theorem}

\newtheorem{remark}{Remark}
\newtheorem{definition}{Definition}
\newcommand{\be}{\begin{eqnarray}}
\newcommand{\ee}{\end{eqnarray}}
\newcommand{\bee}{\begin{eqnarray*}}
\newcommand{\eee}{\end{eqnarray*}}
\newcommand{\R}{{\mathbb R}}

\newcommand{\N}{{\mathbb N}}
\newcommand{\Z}{{\mathbb Z}}

\newcommand{\W}{{W}}
\newcommand{\Hh}{{H}}
\newcommand{\I}{\mathds 1} 

\newcommand{\asy}{{\mathcal O}}
\newcommand{\aasy}{\tilde {\mathcal O}}
\newcommand{\da}{d_A}
\newcommand{\Vector}[1]{{\mathbf {#1}}}
\newcommand{\hhbar}{h}
\newcommand{\Chi}{\xi}
\newcommand{\w}{u}
\newcommand{\UU}{w}
\newcommand{\uu}{\Theta}

\newcommand{\rev}{}

\begin{document}

\title [NLS Stark-Wannier] {Nonlinear Stark-Wannier equation}

\author {Andrea SACCHETTI}

\address {Department of Physics, Informatics and Mathematics, University of Modena e Reggio Emilia, Modena, Italy.}

\email {andrea.sacchetti@unimore.it}

\date {\today}

\thanks {This paper is partially supported by GNFM-INdAM. \ I deeply thank R. Fukuizumi for useful discussions about 
nonlinear Schr\"odinger equations.}

\begin {abstract} 
In this paper we consider stationary solutions to the nonlinear one-dimensional Schr\"odinger equation with a periodic potential and a Stark-type 
perturbation. \ In the limit of large periodic potential the Stark-Wannier ladders of the linear equation become a 
dense energy spectrum because a cascade of bifurcations of stationary solutions occurs when the ratio between the effective nonlinearity strength and 
the tilt of the external field increases.

\bigskip

{\it Ams classification (MSC 2010):} 35Q55, 81Qxx, 81T25. 

\end{abstract}

\maketitle

\section {Introduction} \label {Sec0}

The dynamics of a quantum particle in a periodic potential under an homogeneous external field is one of 
the most important problems in solid-state physics and, more recently, in the theory of Bose Einstein Condensates (BECs). \ Because 
of the periodicity of the potential, it is expected the existence of 
families of stationary (metastable) states with associated energies displaced on regular ladders, the so-called 
Stark-Wannier ladders \cite {GKK,WS,S}, and the wavefunction would perform Bloch oscillations. 

Quantum dynamics becomes more interesting when we take into account the interaction among particles. \ In fact, in the framework of BECs 
accelerated ultracold atoms moving in an optical lattice 
\cite {Bloch1,Bloch2,RSN,SPSSKP,Shin} has opened the field to multiple applications, as well as the measurement 
of the value of the gravity acceleration $g$ using ultracold Strontium atoms confined in a vertical 
optical lattice \cite {FPST,PWTAPT}, direct measurement of the universal Newton gravitation constant 
$G$ \cite {RSCPT} and of the gravity-field curvature \cite {RCSMPT}. 

Motivated by such physical applications we study, as a model for a confined accelerated BECs in a periodic optical lattice under 
the effect of the gravitational force, the nonlinear one-dimensional time-dependent Schr\"odinger equation with a cubic nonlinearity, a  
periodic potential $V$ and an accelerating Stark-type potential $W$ 
\be
i \hbar \frac {\partial \psi}{\partial t} = - \frac {\hbar^2}{2m} \frac {\partial^2 \psi}{\partial x^2}  + 
\frac {1}{\epsilon} V \psi + \alpha_1 \W \psi + \alpha_2 |\psi |^{2 } \psi \, ,  \label {F1}
\ee
in the limit of large periodic potential, i.e. $0< \epsilon \ll 1$; that is equation (\ref {F1}) is the so called Gross-Pitaevskii equation. \ Here, $\hbar$ is 
the Planck's constant, $m$ is the mass of the atom and 
$\alpha_2$ is the strength of the nonlinearity term; the real valued parameters $m$, $\hbar$, $\alpha_1$ and 
$\alpha_2$ are assumed to be fixed. \ In particular $W(x)$ is a Stark-type potential with strength $\alpha_1$, 
that is it is locally a linear function: $W(x)=x$ for any $x$ belonging to a 
fixed interval large enough.  

We name equation (\ref {F1}) \emph {nonlinear Wannier-Stark equation}. \ The well-known \emph {Wannier-Stark equation}, where $\alpha_2 =0$, has been extensively 
studied since the papers by Bloch \cite {BlochF} and Zener \cite {Zener}. \ Assuming that the periodic potential $V$ is regular enough, 
then the spectrum of the associated operator covers the whole real axis. \ On the other side, if we neglect the coupling term between different bands, then it turns out 
that the spectrum of such a decoupled band approximation consists of a sequence on infinite ladders of real eigenvalues \cite {W1,W2}. \ The crucial point is to understand 
what happen to these eigenvalues when we restore the interband coupling term \cite {W3,Z1,Z2}. \ This question has been largely debated and it has been proved that these
ladders of real eigenvalues will turn into ladders of quantum resonances, the so-called Wannier-Stark resonances (see \cite {S} and the references therein). \ Analysis of 
the nonlinear Wannier-Stark equation, where $\alpha_2 \not= 0$, is a completely open problem and it is motivated by recent experiments of BECs in accelerating optical lattices.

By means of a simple recasting we swap the limit of large potential $\epsilon \ll 1$ to a semiclassical equation (see eq. (\ref {F2}) below) where the strength of 
the Stark-type potential and the nonlinearity strength will depend on a semiclassical parameter $\hhbar$. \ In the semiclassical limit of $\hhbar \to 0$ we will 
show that the time-independent nonlinear Schr\"odinger equation may be approximated by means of a discrete time-independent nonlinear 
Schr\"odinger equation which stationary solutions may be explicitly calculated. \ In particular, a cascade of bifurcations occurs 
when the ratio between the nonlinearity strength and the strength of the Stark-type potential increases; in the opposite situation, 
that is when this ratio goes to zero, we recover a local Wannier-Stark ladders picture.

Existence and computation of stationary solutions to equation (\ref {F1}) has been already considered by \cite {FS2,PSM,PS} when $\alpha_1 =0$; in these 
papers the authors reduce the problem of the existence and calculation of stationary solutions to the one related to a 
discrete nonlinear Schr\"odinger equation. \ In this latter problem has been observed by \cite {ABK} that stationary solutions may bifurcate 
when some parameters of the model assume critical values. \ Here, we extend such analysis to the case where 
an external Stark-type potential is present, that is when $\alpha_1 \not= 0$. \ To this end we must 
introduce some technical assumptions on $W$, that is $W$ must be a locally linear bounded function with compact support; in fact in the case of 
a \emph {true} Stark potential where $W(x)=x$ some basic estimates useful in our analysis don't work because $W$ is not a bounded operator. \ Some results, like the occurrence of a cascade of bifurcations for the discrete nonlinear 
Schr\"odinger equation in the anticontinuous limit has been already announced in a physics-oriented paper \cite {S3} without mathematical details. \ We should also mention 
a recent paper \cite {GP} where bifurcations are observed in rotating Bose-Einstein condensates.

The paper is organized as follows: in \S \ref {Sec2} we introduce the model and we state our assumptions; in \S \ref {A} 
we recall some technical results obtained by \cite {FS2}; in \S \ref {Sec3} we derive 
the discrete nonlinear Schr\"odinger Wannier-Stark equation; in \S \ref {Sec4} we compute the finite-mode stationary solutions of the 
discrete nonlinear Schr\"odinger Wannier-Stark equation in the anticontinuous limit, it turns out that a bifurcation tree picture 
occurs; in \S \ref {Sec5} we prove the stability of these stationary solutions when we recover the discrete nonlinear Schr\"odinger 
Wannier-Stark equation; finally, in \S \ref {Sec6}-\ref {Sec7}  we prove that stationary solutions to the complete equation (\ref {F5}) can be 
approximated by means of the finite-mode solutions derived in \S \ref {Sec4}.  

\subsection* {Notation} By $\ell^p_{\R}$ we denote the space of vectors $\Vector {c} = \{ c_n \}_{n\in \Z} \in \ell^p (\Z)$ such that 
$c_n \in \R$ are real valued. \ Similarly, 
\bee
L^p_{\R} = \left \{ \psi \in L^p \ : \ \psi \ \mbox {is a real valued function} \right \}.
\eee

Let $f$ and $g$ two vectors belonging to a normed space with norm $\| \cdot \|$, and depending on the semiclassical parameter 
$\hhbar$. \ By the notation $f = g + \aasy \left ( e^{-S_0/\hhbar} \right )$, as $\hhbar \to 0$, we mean that for any 
$\rho \in (0, S_0)$ there exist a positive constant $C:= C_\rho >0$ (independent of $\hhbar$) such that 
\bee
\| f-g \| \le C e^{-(S_0-\rho )/\hhbar} \, , \  \forall \hhbar \in (0,\hhbar^\star )\, ,
\eee
for some $\hhbar^\star >0$. \ By the notation $f \sim g$, as $\hhbar \to 0$, we mean that $\lim_{\hhbar \to 0^+} \frac {f}{g} = C$ 
for some $C \in (0, +\infty )$. \ By the notation $f = \asy (\hhbar^q )$, as $\hhbar \to 0$, we mean that there exists $\hhbar^\star >0$ and a 
positive constant $C$ independent of $\hhbar$ such that $|f |\le C \hhbar^q$ for any $\hhbar \in (0, \hhbar^\star )$.

By $C$ we denote a generic positive constant independent of $\hhbar$ whose value may change from line to line.

\section {Description of the model and assumptions} \label {Sec2}

Here we consider the nonlinear Schr\"odinger equation (\ref {F1}) where the following assumptions hold true.

{\bf Hyp.1} {\it $V(x)$ is a smooth, real-valued, periodic and non negative function with period $a$, i.e. 
\bee
V(x) = V(x+a) \, , \ \forall x \in \R \, , 
\eee
and with minimum point $x_0 \in \left [ - \frac 12 a ,+\frac 12 a \right )$ such that 
\bee
V(x) > V(x_0) \, , \ \forall x \in \left [ - \frac 12 a ,+\frac 12 a \right ) \setminus \{ x_0 \} \, . 
\eee
For argument's sake we assume that $V(x_0 )=0$ and $x_0=0$.}

In the following let us denote by $x_n = x_0 + n a$.

\begin {remark}
In physical experiments \cite {FPST,PWTAPT} on accelerated BECs in optical lattices the periodic potential has the form 
$V(x) = V_0 \sin^2 (k_L x)$ for some $V_0 ,\ k_L >0$; hence $V(x)$ has a unique minimum point $x_0 =0$ in the interval 
$ \left [ - \frac {\pi}{2 k_L} ,  +  \frac {\pi}{2 k_L} \right )$. \ However, we could, in principle, adapt our treatment 
to a more general case where $V(x)$ has more than one absolute minimum point in such interval. 
\end {remark}

{\bf Hyp.2} {\it $\W(x)$ is a smooth real-valued function such that
\bee
\W (x) = x \ \mbox { if } \ |x|\le N a \, , 
\eee
for some $N \in \N$. \ Furthermore $W$ has compact support $\Omega \supset [-N a, N a].$}

\begin {remark}
We require that $W(x)$ is a bounded function with compact support for technical reasons; indeed, this assumption will play a crucial role 
in order to prove the results given in \S \ref {Sec5}, \ref {Sec6}, \ref {Sec7}. \ However, in practical experiments \cite {FPST,PWTAPT} 
on accelerated BECs in optical lattices it is expected that BECs perform Bloch oscillations in a finite region; hence, a model where the external field $W$ 
has a compact support and it is locally linear in the finite region where Bloch oscillations occur would fit the physical device.
\end {remark}

By recasting 
\be
F = \epsilon \alpha_1 \, , \ h = \hbar \sqrt {\epsilon /2m}  \, , \ \tau = t /\sqrt {\epsilon /2m} \ \mbox { and } \ 
\eta = \epsilon \alpha_2 \label {res}
\ee
then the above equation takes the form
\be
i \hhbar \frac {\partial \psi}{\partial \tau} = - {\hhbar^2} \frac {\partial^2 \psi}{\partial x^2}  +  V \psi + F \W \psi + 
\eta |\psi |^{2} \psi \label {F2}
\ee
and the limit of large periodic potential $\epsilon \to 0^+$ is equivalent to the semiclassical limit $\hhbar \to 0^+$ where
\be
\eta \sim F \sim \hhbar^2 \ \mbox { as } \ \hhbar \ \mbox {goes to zero}. \label {F3}
\ee

We recall here some results by \cite {C1,C2,C3} concerning the solution to the time-dependent nonlinear Schr\"odinger equation 
(\ref {F2}). \ Let $H_B$ be the Bloch operator formally defined on $L^2 (\R , dx)$ as
\be
H_B:=  - {\hhbar^2}\frac {d^2}{dx^2} + V \, . \label {F4}
\ee
For any $N \in \N $, $N>0$, the linear operator $\Hh $, formally defined as 
\bee
\Hh = H_B + F \W 
\eee
on the Hilbert space $L^2 (\R , dx)$, admits a self-adjoint extension, still denoted by $\Hh$. \ The following 
estimate 
hold true (see Proposition 2.1 by \cite {C3}): let $(q,r)$ be an admissible pair $\frac {2}{q} = \frac 12 - \frac 1r$ with 
$2\le q,r \le +\infty$. \ Let $T>0$, then there exists $C:=C(q,T,\hhbar)$ such that 
\bee
\left \| e^{-i \tau \Hh/\hhbar} \psi \right \|_{L^q ([-T,T]; L^r (\R ))} \le C \| \psi \|_{L^2 (\R)} \, , \ 
\forall \psi \in L^2 (\R ) \, . 
\eee
In order to discuss the local and global existence of solutions to (\ref {F2}) \cite {C3} \rev {introduced the following set in a more general situation where 
the potential is not bounded} 
\bee
\Sigma = \left \{ \psi \in H^1 (\R ) \ : \ \| \psi \|_{\Sigma } :=  \| \psi \|_{H^1 (\R )} + \| (V+FW) \psi \|_{L^2 (\R )} < \infty \right \} 
\, .
\eee
%
Then (see Theorem 4.2 by \cite {C3}), if $\psi_0 \in \Sigma$ there exists a unique solution $\psi 
\in C([-T,T]; \Sigma)$ to (\ref {F2}) with initial datum $\psi_0$, such that 
\bee
\psi ,  \psi \partial_x (V+FW), \partial_x \psi \in L^{8} ([-T,T];L^{4} (\R )) \, , 
\eee
for some $T>0$ depending on $\| \psi_0 \|_\Sigma$. \ \rev {We must underline that in our case $\Sigma \equiv H^1 (\R )$ because $V$ and $W$ are 
bounded functions.}

In fact, this solution is global in time for any $\eta \in \R$ because $1 < 2/d$, where $d=1$ is the spatial dimension, and (\ref {F2}) 
enjoys the conservation of the mass
\bee
\| \psi (\cdot , \tau )\|_{L^2 (\R)} = \| \psi_0 (\cdot )\|_{L^2 (\R)}
\eee
and of the energy
\bee
{\mathcal E} (\psi ) = {\mathcal E} (\psi_0 )
\eee
where
\bee
{\mathcal E} (\psi ) &:=& \langle \Hh \psi , \psi \rangle + \frac {\eta}{2} \| \psi 
\|_{L^{4}}^{4} \\ 
&=& {\hhbar^2} \| \partial_x \psi \|_{L^2 (\R)}^2 + \langle V \psi , \psi \rangle + 
F \langle \W \psi , \psi \rangle + \frac {\eta}{2} \| \psi \|_{L^{4}}^{4}
\eee
We may remark that such results hold true even when the Stark-type potential is replaced by an actual Stark potential, i.e. $W(x) \equiv x$. \ In such a 
case $\Sigma \subset H^1 (\R )$.

Here, we look for stationary solutions to equation (\ref {F2}) of the form 
\bee
\psi (x,\tau ) = e^{-i\lambda \tau /\hhbar} \psi (x) 
\eee
for some \emph {energy} $\lambda \in \R$ and wave function $\psi (x)$. \ Hence, equation (\ref {F2}) takes the form
\be
\Hh  \psi + \eta |\psi |^{2 } \psi = \lambda \psi \, .  \label {F5}
\ee

\begin {remark} \label {Rem1}
We must underline that when a stationary solution $\psi$ to equation (\ref {F5}) is regular enough then $\psi$ is, up to a phase factor, a real-valued function 
(see Lemma 3.7 by \cite {P} adapted to (\ref {F5})). \ Hence, equation 
(\ref {F5}) can be replaced by the following equation
\be
\Hh \psi + \eta \psi^{3} = \lambda \psi \, .  \label {F6}
\ee
where $\psi$ is real-valued.
\end {remark}

Our aim is to look for real-valued stationary solutions $\psi \in H^1$ to (\ref {F6}) with associated energy $\lambda \in \R$.

\begin {remark} \label {Rem2}
Let $\left ( T_a \psi \right ) (x) = \psi (x - a)$ be the translation operator. \ Since $[H_B , T_a] = 0 $ and 
$[Fx,T_a]=F a$ then the stationary solutions to (\ref {F6}) when $W$ is a Stark potential, i.e. $W (x) \equiv x$, have associated energies $\lambda$ 
displaced on regular ladders; that is, if $\psi (x) $ is a solution to (\ref {F6}) associated with $\lambda$, then $\psi (x-a)$ is a 
solution to the same equation associated with $\lambda + Fa$. \ From this fact we expect that, under 
some circumstances, the dominant term of the energies $\lambda$ associated to stationary solutions to (\ref {F6}) are displaced on ladders for 
some range of values of $\lambda$, even when $W(x)$ is a Stark-type potential satisfying Hyp.2.
\end {remark}

\section {Preliminary results. \ Bloch functions in the semiclassical limit} \label {A}

\subsection {Bloch Decomposition and Wannier functions} \label {A.1}

Here, we briefly resume some known results by \cite {C,RS} concerning the spectral properties of the self-adjoint 
realization, still denoted by $H_B$, of the Bloch operator formally defined on $L^2 (\R , dx)$ as (\ref {F4}). \ Its spectrum 
is given by bands. \ Let ${\mathcal B} = \R / b\Z = \left ( - \frac 12 b, + \frac 12 b \right ]$, where $b= \frac {2\pi}{a}$ and $a$ 
is the period of the periodic potential $V$, be the Brillouin zone, the elements of the Brillouin zone are denoted by $k$ and they 
are usually named quasi-momentum (or crystal momentum) variable.

Let $\varphi_l (x,k)$ denote the Bloch functions associated to the band functions $E_l (k)$, $l\in \N$. \ Here, we 
collect some basic properties about the Bloch and band functions. \ The band and Bloch functions satisfy to the following 
eigenvalues problem 
\be
H_B \varphi = E \varphi \label {F50}
\ee
with quasi-periodic boundary conditions 
\bee
\varphi (a,k) =e^{i ka} \varphi (0,k) \ \mbox { and } \ \frac {\partial \varphi (a,k)}{\partial x} =e^{i ka} 
\frac {\partial \varphi (0,k)}{\partial x} \, . 
\eee
The Bloch functions $\varphi_l$ may be written as
\bee
\varphi_l (x,k) = e^{i k x} \uu_l (x,k) 
\eee
where $\uu_l (x,k)$ is a periodic function with respect to $x$: $\uu_l (x+a , k) = \uu_l (x,k)$. \ For any fixed $k \in 
{\mathcal B}$ the spectral problem (\ref {F50}) has a sequence of real eigenvalues
\bee
E_1 (k) \le E_2 (k) \le \cdots \le E_l (k) \le \cdots \, ,
\eee
such that $\lim_{l\to \infty} E_l (k) = + \infty$. \ As functions on $k$, both Bloch and band functions are periodic 
with respect to $k$:
\bee
E_l (k) =E_l (k+b) \ \mbox { and } \ \varphi_l (x,k)=\varphi_l (x,k+b) \, ,
\eee
and they satisfy to the following properties for any real-valued $k$:
\bee
\varphi_l (x,-k)= \overline {\varphi_l (x,k)} \ \ \mbox { and } \ \ {E}_l (-k) = {E}_l (k) \, . 
\eee
Furthermore, if $V(x)$ is an even potential, i.e. $V(-x)=V(x)$, then $\varphi_l (-x,k)= \overline 
{\varphi_l (x,k)}$, $\varphi_l (x,0)$ are even functions while $\varphi_l (x, b/2)$ are odd functions. \ The band 
functions $E_l (k)$ are monotone increasing (resp. decreasing) functions for any $k \in \left [ 0 , 
\frac 12 b \right ]$ if the index $l$ is an odd (resp. even) natural number. \ The spectrum of $H_B$ is purely 
absolutely continuous and it is given by bands:
\bee
\sigma (H_B) = \cup_{l\in \N} [E^b_l , E^t_l] \ \ \mbox { where } \ \ [E^b_l , E^t_l] = \{ E_l (k) ,\ 
k \in {\mathcal B} \} \, . 
\eee
In particular we have that 
\bee
E^b_l = 
\left \{
\begin {array}{ll}
E_l (0) & \ \mbox { for odd } l \\ 
E_l (b/2) & \ \mbox { for even } l
\end {array}
\right. \ \mbox { and } \ 
E^t_l = 
\left \{
\begin {array}{ll}
E_l (b/2) & \ \mbox { for odd } l \\ 
E_l (0) & \ \mbox { for even } l
\end {array}
\right. \, . 
\eee
The intervals $(E^t_l, E^b_{l+1})$ are named gaps; a gap $(E^t_l, E^b_{l+1})$ may be empty, that is $E^b_{n+1}=
E^t_l$, or not. \ It is well known that, in the case of one-dimensional crystals, all the gaps are empty if, and 
only if, the periodic potential is a constant function. \ Because we assume that the periodic potential is not a 
constant function then one gap, at least, is not empty. \ In particular when $\hhbar$ is small enough then we have 
that the following asymptotic behavior \cite {H,WK1,WK2}
\be
\frac {1}{C} \hhbar \le E_2^b - E_1^t \le C \hhbar \label {F51}
\ee
holds true for some $C>0$; hence,the first gap between $E_1^t$ and $E_2^b$ is not empty in the semiclassical limit. \ Furthermore, 
the first band turns out to be exponentially small, i.e. 
\be
E_1^t -E_1^b = \aasy ( e^{-C/\hhbar } ) \mbox { for some }C>0; \label {F52}
\ee
in (\ref {F55}) we will give an expression for such a constant $C$.

The Bloch functions are assumed to be normalized to $1$ on the interval $[0,a]$: 
\bee
\frac {2\pi}{a} \int_0^a \overline {\varphi_j (x,k)} \varphi_l (x,k)  d x =\delta_j^l \, ,
\eee
where $\delta_j^l =1$ when $j =l$ and $\delta_j^l =0$ when $j \not= l$ (see Eq. (4.1.8) by 
\cite {C}). \ Furthermore, the Bloch functions are such that (see Eq. (4.1.6a) by \cite {C})
\bee
\int_{\R} \overline {\varphi_j (x,k)} \varphi_l (x,q) dx = \delta_j^l \delta (k-q) 
\eee
and (see Eq. (4.1.10) by \cite {C})
\bee
\sum_{l\in \N} \int_{\mathcal B} \overline {\varphi_m (x,k)} \varphi_l (x',k) dk = \delta (x-x') \, , 
\eee
where $\delta (\cdot )$ denotes the Dirac's $\delta$. \ From the Bloch decomposition formula it follows that any 
vector $\psi \in L^2$ can be written as (see Eq. (5.1.5) by \cite {C} or Theorem XIII.98 by \cite {RS})
\bee
\psi (x) = \sum_{l \in \N}  \int_{{\mathcal B}} \varphi_l (x,k ) \phi_l (k ) d k  \, . 
\eee
The family of functions $\{ \phi_l (k ) \}_{l\in \N}$ is called the crystal momentum 
representation of the wave function $\psi$  and it is defined as
\bee
\phi_l (k ) = \int_{\R} \overline {\varphi_l (x, k )} \psi (x)  d x \, . 
\eee
By construction any function $\phi_l (k)$ is a periodic function and the transformation 
\be
\psi \in L^2 (\R , dx) \to {\mathcal U} \psi := \{ \phi_l \}_{l \in \N} \in {\mathcal H} := 
\otimes_{l\in \N} L^2 ({\mathcal B} , d k ) \label {F53}
\ee
is unitary:
\bee
\| \psi \|_{L^2 (\R , dx )}^2 = \sum_{l\in \N} \| \phi_l \|_{L^2 ({\mathcal B}, dk )}^2\, .
\eee
Let $W_l(x)$ be the \emph {basic} Wannier function associated to the $l$-th band, that is 
\bee
W_l (x) = \sqrt {\frac {a}{2\pi}} \int_{\mathcal B} \varphi_l (x,k) dk \, . 
\eee
We define a family of Wannier functions $\{ W_{l,n} (x) \}_{l\in \N , n \in \Z}$ as 
\bee
W_{l,n} (x) = W_l (x- n a) = \sqrt {\frac {a}{2\pi}} \int_{\mathcal B} \varphi_l (x,k) e^{-i n a k} dk  \, .
\eee
Basically, in the semiclassical limit of $\hhbar$ small, the Wannier function $W_{l,n}$ is localized on the $n$-th 
well, that is in a neighborhood of $x_n$. \ The following properties hold true
\bee
\int_{\R} \overline {W_{l,n} (x)} W_{m, n} (x) d x = \delta_l^m \, , \ \sum_{l\in \N \, , \ n \in \Z} 
\overline {W_{l,n} (x)} W_{l,n} (x') = \delta (x-x') 
\eee
and we have the following relation between the Wannier and the Bloch functions:
\bee
\varphi_l (x,k) = \sqrt {\frac {a}{2\pi}} \sum_{n\in \Z} e^{i n a x} W_{l,n} (x)\, . 
\eee
If we set 
\bee
c_l^n = \int_{\R}  \overline {W_{l,n} (x) } \psi (x) dx 
\eee
then we may represent a wave function $\psi$ as 
\bee
\psi \in L^2 \to {\mathcal W} \psi = \{ c_l^n\}_{l\in \N ,\ n \in \Z} \in \ell^2 (\N \times \Z )
\eee
Such a transformation ${\mathcal W}$ is unitary 
\bee
\| \psi \|_{L^2}^2 = \sum_{l\in \N \, , \ n \in \Z}  | c_l^n |^2 
\eee
with inverse
\be
\psi (x)= \sum_{l\in \N \, , \ n \in \Z}  c_l^n  W_{l,n} (x) \, . \label {F54}
\ee

\begin {remark}
The standard ``tight binding'' model is obtained by substituting (\ref {F54}) in (\ref {F2}), and it reduces 
(\ref {F2}) to a discrete nonlinear Schr\"odinger equation. \ In fact, in order to improve the estimate of the 
remainder terms of the discrete nonlinear Schr\"odinger equation we decompose the wave 
function $\psi (x)$ on a different base where the vectors of such a base are obtained by means of the single 
well semiclassical approximation described in \S \ref {A.2}.
\end {remark}

\subsection {Semiclassical construction} \label {A.2} Here we restrict our attention to just one band, say 
the first one $[E_1^b , E_1^t ]$. \ By assuming $\hhbar $ small enough then the gap between the first band 
and the remainder of the spectrum is open, see equation (\ref {F51}). \ Let $\Pi$ be the spectral projection of 
$H_B$ on the first band; by \cite {Car} we can find a ``good'' orthonormal basis $\{ \w_n \}_{n\in \Z}$ of 
$\Pi \left [ L^2 (\R) \right ]$. 

In one dimension let 
\bee
d_A (x,y) = \int_x^y \sqrt {V(q)} dq 
\eee
be the Agmon distance between $x$ and $y$ (associated to the energy level corresponding to the minimum value $V (x_0 )=0$ 
of the potential $V(x)$) and let
\be
S_0 = d_A (x_n , x_{n+1}) \label {F55}
\ee
be the Agmon distance between two adjacent wells; by periodicity of the potential $V(x)$ then $S_0$ is independent of the 
index $n$. 
%

Here we summarize some important properties of $\{ \w_n \}_{n\in \Z}$ (see \cite {Car} and Appendix A by \cite {FS2}). \ Let 
$\tilde V$ be the ``single well potential'' obtained by filling all the well, but one; 
that is $\tilde V(x) = V(x) + \theta (x)$ where $\theta (x)$ is a smooth and non negative function such that 
$\theta (x) =0$ in a small neighborhood $(x_0 - \delta ,x_0 + \delta )$ of $x_0$ and $\theta (x) > \varepsilon$ for 
any $x \notin (x_0 - 2\delta ,x_0 + 2\delta )$ for some $\varepsilon >0$ and $0<\delta < \frac 14 a$ is fixed. \ Then 
the operator $\tilde H = -  {\hhbar^2} \frac {d^2}{dx^2} + \tilde V$ has discrete spectrum in the interval 
$[0,\varepsilon]$ and we call such eigenvalues single well states. \ We denote by $\Lambda_1$ the first one, the so 
called ``single well ground state'', and by $\UU_0(x)$ the associated eigenvector. \ 

\begin {remark} \label {Rem16}
By means of semiclassical arguments it follows that \cite {H,WK1,WK2}
\bee
\mbox {dist} \left ( \Lambda_1 , [E_1^b, E_1^t] \right ) = \aasy \left ( e^{-S_0/\hhbar} \right ) \, . 
\eee
Furthermore,
\bee
E^b_2 - \Lambda_1 \ge C \hhbar  
\eee
for some $C>0$.
\end {remark}

If we denote $\UU_n (x) = \UU_0 (x-na)$ then the family $\{ \UU_n \}_{n\in \Z}$ is a family of linearly independent vectors localized on 
the $n-$th well. \ Then, taking their projection $\Pi \UU_n$ on $\Pi \left [ L^2 (\R ) \right ]$ and orthonormalizing the obtained 
family we finally get the base $\{ \w_n \}_{n\in \Z}$ of $\Pi \left [ L^2 (\R ) \right ]$. 

\begin {lemma} \label {Lemma8} The vectors $\w_n$ of the orthonormal base of $\Pi [L^2 (\R )]$ are such that:
\begin {itemize}
\item [i.] 
The matrix with real-valued elements $\langle \w_m , H_B \w_n \rangle$ can be written as 
\bee
(\langle \w_m , H_B \w_n \rangle ) = \Lambda_1 \I - \beta {\mathcal T} + \tilde D, 
\eee
where ${\mathcal T}$ is the tridiagonal Toeplitz matrix, i.e., 
\bee
({\mathcal T})_{m,n} = 
\left \{
\begin {array}{ll}
0 & \ { if }\ |m-n| \not= 1 \\
1 & \ { if }\ |m-n| = 1
\end {array}
\right. \, ,
\eee
$\beta >0$ is such that for any $\rho >0$ then 
\be
\frac 1C e^{-(S_0+\rho )/\hhbar} \le \beta \le C e^{-(S_0-\rho )/\hhbar}  \label {F55Bis}
\ee
for some positive constant $C:=C_\rho >0$, and the remainder term $\tilde D$ is a bounded linear operator from 
$\ell^p (\Z)$ to $\ell^p (\Z)$ with bound 
\be
\| \tilde D \|_{{\mathcal L} (\ell^p  \to \ell^p )} \le C e^{-(S_0 + \zeta )/\hhbar }\, , \ p \in 
[1,+\infty ] \, , \label {F56}
\ee
for some positive constant $\zeta >0$ independent of $\hhbar$ and $p$, and for some positive constant $C$ 
which depends only on $p$. 

\item [ii.] Let $T_a$ be the translation operator $\left ( T_a \psi \right ) (x) = \psi (x-a)$, where $a$ is the 
period of $V$. \ Then, $\w_n =T_a^n \w_0$. 

\item [iii.] All the functions $\w_n$ can be chosen to be real-valued by means of a suitable gauge choice. 

\item [iv.] For any $\rho' , \ \rho'' >0$ and for some positive constant $C>0$ independent on the indexes $n$ 
and $m$, we have that 
\bee
 \| \w_m \w_n \|_{L^1} \le C e^{-[(S_0-\rho' )|m-n|- \rho'']/\hhbar } \, , \ m\not= n \, .
\eee

\item [v.] There exists a constant $C>0$ independent of $\hhbar$ such that 
\bee
\left \| \sum_{n\in \Z} |\w_n| \right \|_{L^\infty} \le C \hhbar^{-1/2} \, . 
\eee

\item[vi.] For any $p\in [2,\infty]$, $\|\w_n\|_{L^p} \le C \hhbar^{-{(p-2)}/{4p}}$, and $\left \| \frac {d \w_n}{dx} \right \|_{L^2} 
\le C \hhbar^{-{1}/{2}},$ where the constants $C>0$ are independent of $\hhbar$ and $n$.

\end {itemize}

\end {lemma}

\section {Construction of the discrete nonlinear Stark-Wannier equation} \label {Sec3}

Let $\Pi$ the projection operator associated to the first band $[E_1^b, E_1^t ]$ of $H_B$ (see \S \ref {A.1}) and let $\Pi_\perp =\I - \Pi$. \ Let 
\be
\psi = \psi_1 + \psi_\perp \ \mbox { where } \psi_1 = \Pi \psi \ \mbox { and } \ \psi_\perp = \Pi_\perp 
\psi \, . \label {F7}
\ee
By the Carlsson's construction resumed in \S \ref {A.2} we may write $\psi_1$ by means of a linear combination of a suitable orthonormal base 
$\{ u_n \}_{n \in \Z}$ of the space $\Pi \left [ L^2 (\R ) \right ]$, that is 
\be
\psi_1 (x) = \sum_{n\in \Z} c_n \w_n (x) \, . \label {F8}
\ee
where $u_n \in H^1 (\R )$ and 
\bee
\Vector {c} = \{ c_n \}_{n \in \Z} \in \ell^2_{\R} (\Z ) 
\eee
because $\psi$, and then $\psi_1$, is a real-valued function by Remark \ref {Rem1} and $\w_n$ are real valued too since Lemma 
\ref {Lemma8}.iii. \ \rev {In fact, when we make use of the fixed point argument in \S 7 and when we prove the existence result of 
stationary solutions in \S 8 we work with vectors $\Vector {c} \in \ell^1 (\R )$; then in the sequence we may assume that $\Vector {c} \in \ell^p (\R )$ 
for any $p \in [1,+\infty ]$.}

\begin {remark} \label {Rem3}
By construction 
\bee
\| \psi_1 \|_{L^p} &= & \left \| \sum_{n\in \Z} c_n \w_n \right \|_{L^p} \le \sum_{n\in \Z} |c_n| \max_n \| \w_n \|_{L^p} \le  
\| \Vector {c} \|_{\ell^1} \| \w_0 \|_{L^p} \\
&\le & C \hhbar^{- {(p-2)}/{4p}} \| \Vector {c} \|_{\ell^1} 
\eee
by Lemma \ref {Lemma8}.ii and Lemma \ref {Lemma8}.vi.
\end {remark}

\begin {remark} \label {Rem4} We must underline that the standard tight-binding model is constructed by making 
use of the Wannier functions (see (\ref {F54})) instead of (\ref {F7}) and (\ref {F8}). \ In fact, the decomposition (\ref {F54}) turns out to be 
more natural and it has the advantage to work for any range of $\hhbar$; decompositions (\ref {F7}) and (\ref {F8}) are more 
powerful than (\ref {F54}) in the semiclassical regime of $\hhbar \ll 1$ and they have the great advantage that the vectors $\w_n$ are 
explicitly constructed by means of the semiclassical approximation (see Lemma \ref {Lemma8}).
\end {remark}

By inserting (\ref {F7}) and (\ref {F8}) in equation (\ref {F6}) then it takes the form
\be
\left \{ 
\begin  {array}{lcl}
\lambda c_n &=& \langle \w_n , H_B \psi \rangle + F \langle \w_n , \W \psi \rangle + \eta \langle \w_n , 
\psi^{3} \rangle \, , \ n \in \Z \\ 
\lambda \psi_\perp &=& \Pi_\perp H_B \psi + F \Pi_\perp \W \psi + \eta \Pi_\perp \psi^{3} 
\end  {array}
\right.  \, , \label {F9}
\ee
where $\Vector {c} \in \ell^2_{\R}$ and $\psi_\perp$ are such that 
\bee
\| \psi \|^2_{L^2} = \| \Vector {c} \|^2_{\ell^2} + \| \psi_\perp \|^2_{L^2} \, . 
\eee

The following result immediately follows by Lemma \ref {Lemma8}.

\begin {lemma} \label {Lemma1} We have that 
\bee
\langle \w_n , H_B \psi \rangle = \Lambda_1 c_n - \beta (c_{n+1} + c_{n-1}) + r_1^n \, , 
\eee
where $\beta$ satisfies (\ref {F55Bis}) and 
\bee
\ r_1^n := \sum_{m\in \Z} \tilde D_{n,m} c_m 
\eee
where $\tilde D_{n,m}$ is defined by Lemma \ref {Lemma8}.i and it satisfies to the following estimate for some $\zeta >0$: let $\Vector {r_1} =\{ r_1^n \}_{n\in \Z}$ 
and $\Vector {c} =\{ c_n \}_{n\in \Z} \in \ell^p_{\R}$, then 
\bee
\| \Vector {r_1} \|_{\ell^p} \le C e^{-(S_0 + \zeta )/\hhbar} \| \Vector {c} \|_{\ell^p} \, ,\ \forall p\in [1,+\infty]\, , 
\eee
for some positive constant $C:=C_p >0$.
\end {lemma}

Let $\da (x,y)$ be the Agmon distance between two points $x,y \in \R$ and let $S_0 := \da (x_n, x_{n+1})$, $n \in \Z$, be the Agmon distance between the bottoms $x_n$ 
and $x_{n+1}$ of two adjacent wells of the periodic potential $V$ (for further details see \S \ref {A.2}); by periodicity $S_0$ does not depend on the index $n$.

\begin {lemma} \label {Lemma2} We have that
\bee
\langle \w_n , \W \psi \rangle = a \tilde   \Chi (n) c_n + r_2^n + r_3^n \, , 
\eee
where for any $\rho >0$ there exists $C:=C_\rho$ such that 
\bee
\| \Vector {r_2} \|_{\ell^1} \le C e^{-(S_0- \rho )/\hhbar } \| \Vector {c} \|_{\ell^1} 
\eee
and there exists $C>0$ such that
\bee
\| \Vector {r_3} \|_{\ell^1} \le   C \| \psi_\perp \|_{L^2} \, . 
\eee
Furthermore, $|\tilde   \Chi (n)| \le C$ for any $n$ because $\W$ is bounded, and
\be
\tilde   \Chi (n) =  \frac {C_0}{a} +   \Chi (n) + \tilde \asy \left ( e^{-S_0/\hhbar } \right )  \label {F11} 
\ee
where $\Chi (n)$ is a bounded function such that 
\bee
\Chi (n) =  n \, \mbox { if } |n|\le N \, , \mbox { and } \ \, C_0 = \int_{a_-}^{a_+} x |\w_0 (x) |^2 dx \,  
\eee
where $a_\pm$ are such that $a_- < x_0=0 <a_+$ and $d_A (a_- , x_0) = d(x_0 , a_+) = \frac 12 S_0$; by 
construction $a_+ - a_- = a$.  
\end {lemma}

\begin {proof} By inserting (\ref {F7}) and (\ref {F8}) in $\langle u_n , W \psi \rangle $ one gets 
\bee
\langle u_n , W \psi \rangle &=& \langle u_n , W \psi_1 \rangle + \langle u_n , W \psi_\perp \rangle \\ 
&=& \langle u_n , W u_n \rangle c_n + \sum_{m\not= n} \langle u_n , W u_m \rangle c_m + \langle u_n , W \psi_\perp \rangle \, . 
\eee
Now, we set 
\bee
r_2^n = \sum_{m\not= n} \langle u_n , W u_m \rangle c_m \, , \ r_3^n = \langle u_n, W\psi_\perp \rangle \ \mbox { and } \ 
\tilde \Chi (n) = \frac {1}{a} \langle u_n , W u_n \rangle \, . 
\eee
Estimates of $\Vector {r}_2 = \{ r_2^n \}_{n\in \Z}$ and $\Vector {r}_3= \{ r_3^n \}_{n\in \Z}$ directly come from the properties collected in 
Lemma \ref {Lemma8}. \ Indeed
\bee
\left | \langle \w_n , \W \w_m \rangle \right | \le \| \W \|_\infty \| \w_n \w_m \|_{L^1} \le 
C e^{-[(S_0 - \rho')|m-n| - \rho'' ]/\hhbar }
\eee
for any $\rho' ,\rho'' >0$ and some $C>0$, because $\W$ is bounded; hence the estimate 
\bee
\| \Vector {r_2} \|_{\ell^1} \le \sum_n \sum_{m\not= n} \left | \langle \w_n,\W \w_m \rangle \right | \, 
|c_m| \le C e^{-(S_0- \rho )/\hhbar } \| \Vector {c} \|_{\ell^1} 
\eee
follows. \ Similarly, 
\bee
\| \Vector {r_3} \|_{\ell^1} &=& \sum_n \left | \langle \w_n , \W \psi_\perp \rangle \right | \le 
\left \langle  \sum_n | \w_n |\chi_\Omega , |\W \psi_\perp | \right \rangle \\ 
 &\le &  C \| \psi_\perp \|_{L^2} 
\eee
where $\chi_\Omega $ is the characteristic function and where $\Omega $ is the compact support of $W$. \ Concerning 
estimate (\ref {F11}) we consider the term $\langle \w_n , \W \w_n \rangle$ when $|n| \le N$; let 
\bee
&& \langle \w_n , \W \w_n \rangle = \int_{na+a_-}^{na+a_+} \W (x) 
|\w_n (x)|^2 dx + \\ 
&& \ + \int_{-\infty}^{na+a_-} \W (x) 
|\w_n (x)|^2 dx + \int_{na+a_+}^{+\infty} \W (x) 
|\w_n (x)|^2 dx 
\eee
where
\bee
\int_{na+a_-}^{na+a_+} \W (x) 
|\w_n (x)|^2 dx &=& \int_{na+a_-}^{na+a_+} x 
|\w_n (x)|^2 dx \\ 
&=& C_0 + n a \int_{ a_-}^{ a_+} |\w_0 (y)|^2 dy = C_0 +n a \left [ 1+ \tilde \asy 
\left ( e^{-S_0/\hhbar } \right ) \right ] 
\eee
because $\w_n (y+na) = \w_0 (y)$, Lemma \ref {Lemma8}.ii and Lemmata 4.iii and 7 by \cite {FS2}. \ More precisely, let 
$\Omega_0 = \R \setminus \left [a_- , a_+ \right ] $ then
\bee
\int_{ a_-}^{a_+} |\w_0 (y)|^2 dy 
= 1- \int_{\R} | \chi_{\Omega_0} (y)|^2 |\w_0 (y)|^2 dy 
\eee
where $ \chi_{\Omega_0}$ is the characteristic function on $\Omega_0$. \ Then (the properties below 
concerning $\UU_0$ are given in Lemma 4.iii by \cite {FS2}, where $\UU_0 (x)$ is the single well ground state defined in \S \ref {A.2})
\bee
\|  \chi_{\Omega_0} \w_0 \|_{L^2} &\le &\|  \chi_{\Omega_0} \UU_0 \|_{L^2} + \|  \chi_{\Omega_0} (\w_0-\UU_0) \|_{L^2} 
\le \tilde \asy (e^{-S_0/2\hhbar}) + \tilde \asy (e^{-S_0/\hhbar}) \\ 
&=& \tilde \asy (e^{-S_0/2\hhbar}) \, .
\eee
Hence,
\bee
\int_{a_-}^{a_+} |\w_0 (y)|^2 dy = 1 - \tilde \asy (e^{-S_0/\hhbar}) \, . 
\eee
Concerning the estimate of the remainder terms we have that 
\bee
\left | \int_{-\infty}^{na+a_-} \W (x) 
|\w_n (x)|^2 dx \right | \le C \int_{\R}  \chi_{\left ( -\infty,
na+a_- \right )}^2 (x) |\w_n (x)|^2 dx = \tilde \asy \left ( e^{-S_0/\hhbar} 
\right )  
\eee
because $\W$ is bounded and by making use of the same arguments as before. \ Similarly we get the same estimate 
for $\int^{+\infty}_{na+a_+} \W (x) |\w_n (x)|^2 dx $. 
\end {proof}

\begin {remark}
By construction and since $\w_0$ is normalized to one it follows that $ | C_0 |  \le C$ for some positive constant $C>0$ independent of $\hhbar$.
\end {remark}

Finally, concerning the nonlinear term we recall the following result which follows by \cite {FS2} (where we choose $\sigma =1$, for the purpose of completeness the detailed proof is given in a separate appendix).

\begin {lemma} \label {Lemma3} We have that 
\bee
\langle \w_n , \psi^{3} \rangle = C_1 c_n^{3} + r_4^n \, ,
\eee
where
\bee
C_1 = \| \w_n \|_{4}^{4} \equiv  \| \w_0 \|_{4}^{4} 
\eee
and 
\bee
r_4^n = \left (\langle \w_n , \psi^{3} \rangle - C_1 c_n^{3} \right ) 
\eee
satisfies to the following estimate: let $\Vector {r_4} = \{ r_4^n \}_{n\in \Z}$, then for any $\rho >0$ there exists $C:=C_\rho$ such that 
\bee
\| \Vector {r_4} \|_{\ell^1} \le
C \left [ \hhbar^{-1/2} \| \psi_\perp \|_{H_1}^3 + \| \Vector {c} \|_{\ell^1} \| \psi_\perp \|_{H_1}^2 +  \| \Vector {c} \|_{\ell^1}^2 
\hhbar^{-1/4}  \| \psi_\perp \|_{H^1} + \| \Vector {c} \|^3_{\ell^1} e^{-(S_0-\rho )/\hhbar } \right ] 
\eee
\end {lemma}

\begin {remark} \label {Rem7}
By Lemma \ref {Lemma8}.vi it follows that $C_1 \sim \hhbar^{- {1}/{2}}$ as $\hhbar$ goes to zero.
\end {remark}

Therefore, equation (\ref {F9}) takes the form
\be
\left \{ 
\begin  {array}{lcl}
\lambda c_n &=& (\Lambda_1 + F C_0) c_n - \beta (c_{n+1} + c_{n-1}) + F   \Chi (n) a c_n + 
\eta C_1 c_n^{3}  + r^n \, ,  \\
\lambda \psi_\perp &=& H_B  \psi_\perp + F \Pi_\perp \W \psi + \eta \Pi_\perp \psi^{3} 
\end  {array}
\right. \label {F12}
\ee
where 
\be
r^n = r_1^n + F(r_2^n + r_3^n) + \eta r_4^n + F r_5^n , \ r_5^n := a \left [  \tilde   \Chi (n) -   \Chi (n) - \frac {C_0}{a} \right ] c_n \, . \label {F13} 
\ee

\begin {definition} \label {Definizione1} 
We define the discrete nonlinear Stark-Wannier equation (hereafter DNLSWE) as 
\be
\lambda g_n &=& (\Lambda_1 + F C_0) g_n - \beta (g_{n+1} + g_{n-1}) + F   \Chi (n) a g_n + 
\eta C_1 g_n^{3}  \label {F14}
\ee
where $\Vector {g} =\{ g_n \}_{n\in \Z} \in \ell^2_{\R} (\Z) $.
\end {definition}

As already explained in Remark 2 we expect that the solutions to equation (\ref {F14}) are displaced, when $\Chi (n) \equiv n$ (corresponding to the 
case $W(x) \equiv x$), on regular ladders, that is the solutions $\lambda$ are of the form $\lambda_j = \lambda_0 + j F a$ for some $\lambda_0 \in \R$ and 
any $j \in \Z$. \ We will call, hereafter, the value $\lambda_j$ as the $j$-th \emph {rung} of the ladder connected to $\lambda_0$. \ In the case that $W(x)$ 
is a linear function on an interval $[-N a,N a]$, according with Hyp. 2, then we will see that the structure of the ladder locally occurs, so even in such a 
case we may speak of \emph {rungs} of such a kind of ladders of stationary solutions.

\section {Anticontinuous limit of the DNLSWE} \label {Sec4} Let us set 
\be
\tilde \lambda := \lambda - (\Lambda_1 + F C_0 ) \, , \ \nu := \eta C_1 \, , \ f := F a \label {F13Bis}
\ee
where
\be
f \sim \hhbar^2 \, , \ \nu \sim \hhbar^{ 3/2} \ \mbox { and } \ \beta = \aasy \left ( e^{-S_0/\hhbar} \right ) \label {F15}
\ee
since Remark \ref {Rem7} and Lemma \ref {Lemma8}.i. \ For argument's sake, we assume that $f \, , \ \nu \ge 0$. \ Hence (\ref {F14}) takes the form
\be
\tilde \lambda g_n = - \beta \left ( g_{n+1} + g_{n-1} \right ) + f   \Chi (n) g_n + \nu g^{3}_n \label {F16} \, . 
\ee
and in the \emph {anticontinuous limit} $\beta \to 0$ then (\ref {F16}) becomes
\be
\left (\tilde \lambda - \nu d_n^{2} \right )d_n = f   \Chi (n) d_n \, , \Vector {d} = \{ d_n \}_{n \in \Z} \in \ell^2_\R (\Z )\, .  \label {F17}
\ee

\subsection {Finite-mode solutions to the anticontinuous limit equation (\ref {F17})}

Here, we look for stationary solutions $\Vector {d} \in \ell^2_\R$ to (\ref {F17}) under the normalization condition 
\bee
\| \Vector {d} \|_{\ell^2 }^2 = \sum_{n \in \Z} d_n^2 =1 \, . 
\eee

\begin {definition} \label {Definizione2} 
We say that the anticontinuous limit equation (\ref {F17}), under the normalization condition $\| \Vector {d} \|_{\ell^2} =1$,  
has a one-mode solution if there exists a set $S \subset \Z$, hereafter called solution-set, with finite cardinality,  
a real value $\mu^S$ and a normalized vector $\Vector {d}^S = \{ d_n^S\}_{n\in \Z} \in \ell^2_{\R} (\Z)$ where $\mu^S$ and 
$\Vector {d}^S$ solve 
\be
\left (\mu^S - \nu d_n^{2} \right )d_n = f   \Chi (n) d_n  \, , \ \mbox { with } d_n \not= 0 \mbox { if } n \in S \label {F18}
\ee
and where $d_n^S =0$ if $n\notin S$. \ The real value $\mu^S$ is hereafter called the ``energy'' associated to the 
stationary solution $\Vector {d}^S$.
\end {definition}

When $\nu =0$ then we simple recover a (kind of) Stark-Wannier ladder, that is the solution-sets are given by simple sets of 
the form $S=\{ j \}$ for any $j \in \Z$ and we have a family of admitted ``energies'' 
$\mu^S = f   \Chi (j )$ with associated 
stationary solutions $\Vector {d}  =\left \{ \delta_n^{j} \right \}_{n\in \Z}$. \ In fact, it is an exact Stark-Wannier ladder 
when $\Chi (n) \equiv n$.

Assume now that the effective nonlinearity strength is not zero, that is $\nu >0$ for argument's sake. \ In such a case, equation 
(\ref {F18}) has finite mode solutions ${\bf d}^S = \{ d_n^S\}_{n \in \Z}$, associated to sets $S \subset \Z$ with finite 
cardinality ${\mathcal N} = \sharp S < \infty$, given by 
\be
d^S_n = 
\left \{
\begin {array}{ll}
 0 & \ \mbox { if } n \notin S \\ 
\pm \left [ \frac {\mu^S - f   \Chi (n)}{\nu } \right ]^{1/2 } & \ \mbox { if } n \in S 
\end {array}
\right.  \, , \label {F19}
\ee
with the condition 
\be
\mu^S - f   \Chi (n) > 0 \, , \ n \in S \, , \label {F20}
\ee
because we have assumed that $d_n^S \in \R$ and $\nu >0$. \ The normalization condition reads
\be
1 = \| \Vector {d^S} \|_{\ell^2} = \sum_{n \in S} (d_n^S)^{2 } = \sum_{n \in S} \frac {\mu^S - f   \Chi (n)}{\nu} \, . \label {F21}
\ee
In the case ${\mathcal N}=1$ then $S= \{ j \}$ again for any $j \in \Z$ and (\ref {F21}) reduces to
\bee
\mu^{S} =\nu +  {f  \Chi (j )}
\eee
where condition (\ref {F20}) holds true because we have assumed that $\nu > 0$; the associated 
stationary solution ${\bf d}^{S}$ takes the form:
\bee
d^{S}_n = 
\left \{
\begin {array}{ll}
 0 & \ \mbox { if } n \not= j \\ 
 \pm 1 & \ \mbox { if } n=j 
\end {array}
\right. \, . 
\eee
That is we recover a kind of (perturbed) Stark-Wannier ladder.

\begin {remark}
From this fact we can conclude that the anticontinuous limit (\ref {F17}) always admits a ladder-type family of normalized 
one-mode solutions.
\end {remark}

\subsection {Finite-mode solutions to equation (\ref {F18}) associated to solution-sets $S$ with finite cardinality bigger that 1} 
In order to look for finite-mode solutions 
with ${\mathcal N}>1$ the normalization condition (\ref {F21}) implies that 
\be
\mu^S = \frac {\nu}{{\mathcal N}} + \frac {f}{{\mathcal N}} \sum_{n \in S}   \Chi (n ) \ \mbox { with } \ \max_{n \in S}   \Chi (n ) < 
\frac {\mu^S}{f} \, .  \label {F22}
\ee

\subsubsection {Existence of finite-mode solutions.} Stationary solutions ${\bf d}^S$ associated to the energy (\ref {F22}) are 
given by
\be
d^S_n = 
\left \{
\begin {array}{ll}
 0 & \ \mbox { if } n \notin S \\ 
\pm \left [ \frac 1{\mathcal N} + \frac {f}{\nu {\mathcal N}} \sum_{\ell \in S}   \Chi (\ell ) - \frac {f}{\nu}   \Chi (n) \right ]^{1/2} & \ \mbox { if } n \in S 
\end {array}
\right. \, . \label {F23}
\ee

In the case ${\mathcal N}=2$ let $S= \{ j,j +\ell_1 \}$ with $\ell_1 >0$. \ The eigenvalue equation (\ref {F22}) becomes 
\be
\mu^S =\frac {\nu}{2} + \frac {f}{2} \left [   \Chi (j) +   \Chi (j +\ell_1) \right ]  \label {F24}
\ee
where condition $\max_{n\in S}   \Chi (n) < \frac {\mu^S}{f}$ becomes 
\bee
f   \Chi (j +\ell_1) < \frac 12 \nu + \frac 12 f \left [   \Chi (j) +   \Chi (j + \ell_1) \right ]\, , 
\eee
that is 
\be
0 \le \left [   \Chi (j+\ell_1) -   \Chi (j) \right ] < \frac {\nu}{f} \, . \label {F25}
\ee

In conclusion, if 

\begin {itemize}

\item [-] $ \frac {\nu}{f} \le \left [   \Chi (j+\ell_1) -   \Chi (j) \right ]$ then (\ref {F25}) is not satisfied and there 
are no stationary solutions associated to solution-sets of the form $S =\{ j +\ell_1 , j \}$ with cardinality $2$;

\item [-] $0 \le \left [   \Chi (j+\ell_1) -   \Chi (j) \right ] < \frac {\nu}{f} $ we have a family of two-mode solutions 
associated to solution-sets $S=\{ j , j+\ell_1 \}$ with $\mu^S$ given by (\ref {F24}) and where
\bee
d^S_n = 
\left \{
\begin {array}{ll}
 0 & \mbox { if } n \not= j, j+\ell_1 \\ 
\pm \left [ \frac 12  + \frac 12 \frac {f }{\nu} \left (   \Chi (j +\ell_1) -   \Chi (j) \right ) \right ]^{1/2} & \mbox { if }n=j \\ 
\pm \left [ \frac 12  - \frac 12 \frac {f }{\nu} \left (   \Chi (j +\ell_1) -   \Chi (j) \right ) \right ]^{1/2} & \mbox { if }n=j+\ell_1
\end {array}
\right. 
\eee

\end {itemize}

Finally, we can extend such an argument to any integer number ${\mathcal N} >1$ obtained the following result.

\begin {theorem} \label {Thm1}
Let $S= \{ j+\ell_0, j+\ell_1, \ldots , j + \ell_{{\mathcal N}-1} \}$, with $j \in \Z$ and 
$0=\ell_0< \ell_1 < \ell_2 < \ldots < \ell_{{\mathcal N}-1}$ positive integer numbers 
such that 
\be
  \Chi \left ( j + \ell_{{\mathcal N}-1} \right ) < 
\frac {\nu}{f {\mathcal N}} + \frac {1}{{\mathcal N}} \sum_{k =0}^{{\mathcal N}-1} 
  \Chi \left ( j +  \ell_k \right ) \label {F26}
\ee
Then $S$ is a solution-set connected to the $j$-th \emph {rung} of a (kind of) Stark-Wannier ladder and equation (\ref {F18}) has a 
${\mathcal N}$-mode solution with  
\be
\mu^S = \frac {\nu}{\mathcal N} + \frac {f}{\mathcal N} \sum_{k =0}^{{\mathcal N}-1}   \Chi 
\left ( j + \ell_k \right ) \label {F27}
\ee
and associated normalized stationary solution given by (\ref {F23}).
\end {theorem}

\begin {remark} \label {Rem9}
We should underline that some of such a solution may be associated to the same ``energy'' $\mu^S$. \ For instance let $N>5$ and let us consider the sets 
$S_1=\{ 0, 3, 4\}$ and $S_2 = \{ 0, 2, 5\}$. \ Recalling that $  \Chi $ is a linear function in both sets $S_1$ and $S_2$ then they 
are associated to the same value (where we assume, for argument sake, that $a=1$ and $C_0=0$) of energy 
\bee
\mu = \frac 13 \nu + \frac 73 f \, .
\eee
In Figure \ref {Figura1} - left panel - we plot the 4 solutions (\ref {F23}) corresponding to the set $S_1$. \ In Figure \ref 
{Figura1} - right panel - we plot the solutions (\ref {F23}) with sign $+$, corresponding to the sets $S_1$ and $S_2$.
\end {remark}

\begin{center}
\begin{figure}
\includegraphics[height=5cm,width=5cm]{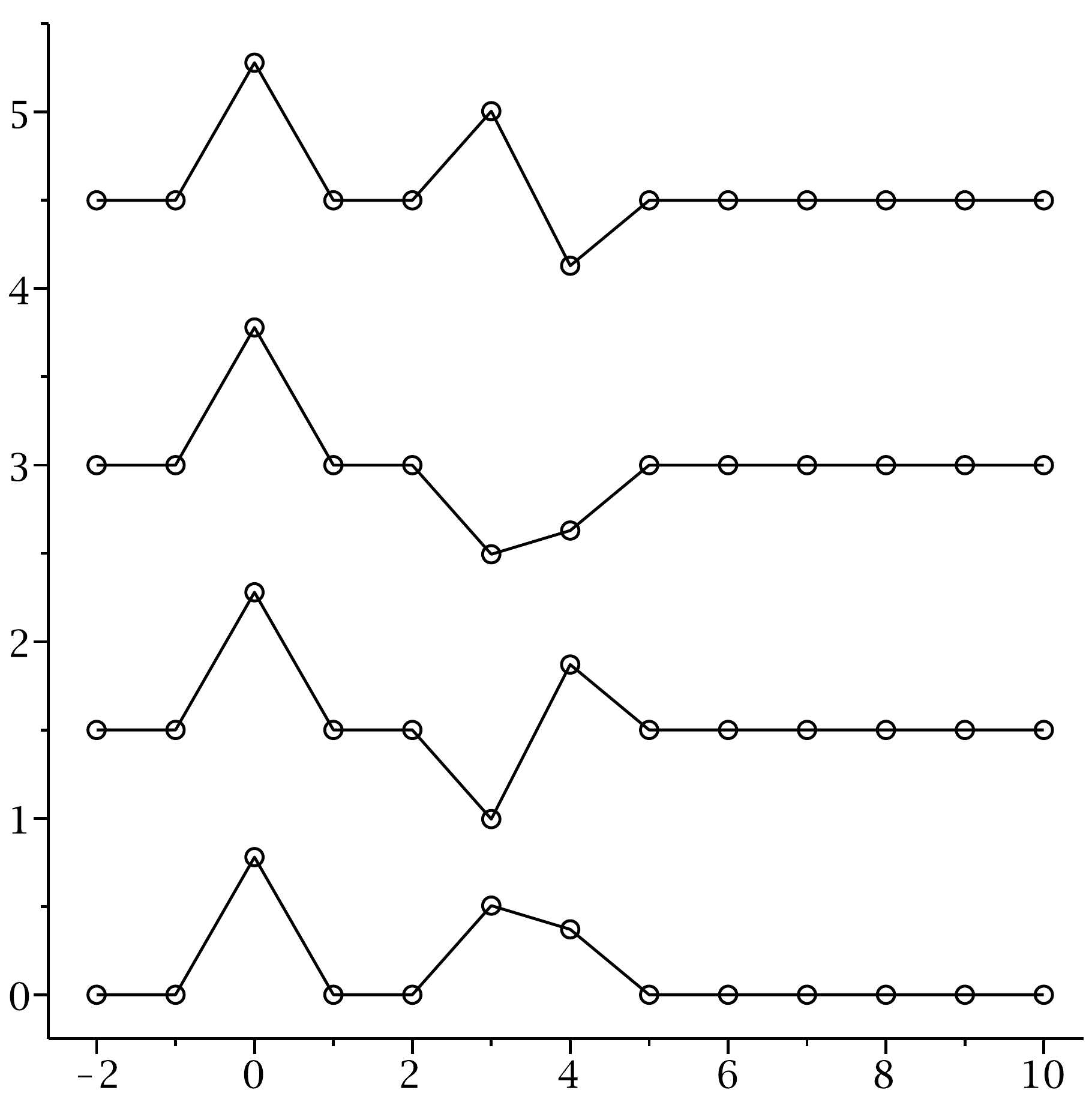}
\includegraphics[height=5cm,width=5cm]{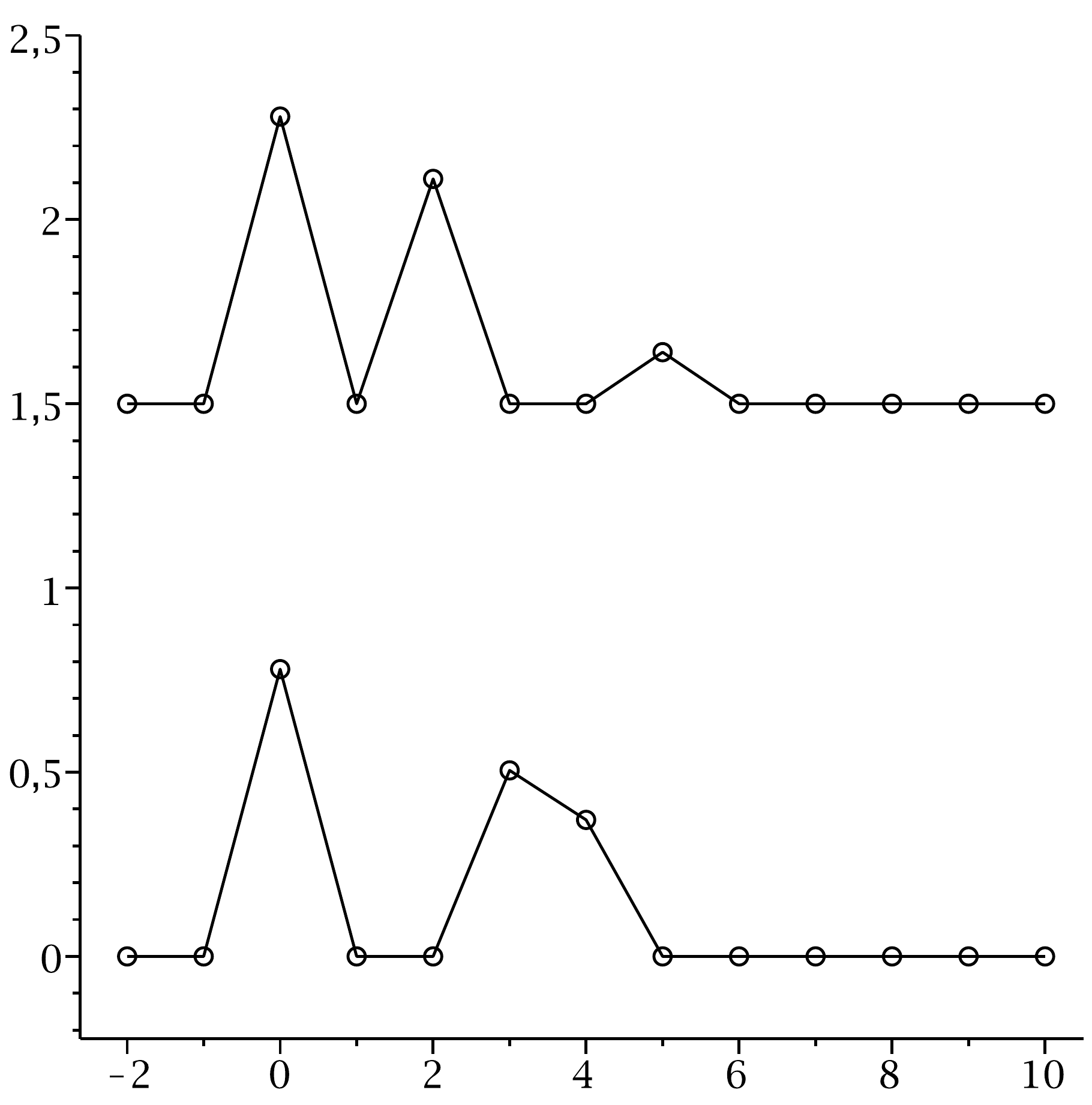}
\caption{\label {Figura1} In the left panel we plot the 4 solutions $\Vector {d}^{S_1}$ corresponding to the solution-set 
$S_1=\{ 0, 3, 4\}$. \ In the right panel we plot the solutions $\Vector {d}^{S_1}$ and $\Vector {d}^{S_2}$ given by (\ref {F23}) 
with sign $+$, corresponding to the solution-sets $S_1$ and $S_2= \{ 0, 2, 5\}$; both solutions are associated to the same value of the energy 
$\mu$.}
\end{figure}
\end{center}

\begin {remark} \label {Rem10}
If the solution-set $S = \{ 0, \ell_1, \ \ldots ,\ \ell_{{\mathcal N}-1} \} \subset [-N,+N]$ then $ \Chi_n (\ell )$, $\ell \in S$, 
is linear and thus we locally recover a Stark-Wannier ladder structure. \ That is $S'=\{ j, j+\ell_1, \ \ldots ,\ j+\ell_{{\mathcal N}-1} \}$ 
is a solution-set too, provided that $|j|, \, |j+\ell_{{\mathcal N}-1} | \le N$, and $\mu^{S'} - \mu^S = j f$. \ We will say 
that $\mu^S$ is connected with the $0$-th \emph {rung} of the ladder, and that  $\mu^{S'}$ is connected with the $j$-th 
\emph {rung} of the ladder.
\end {remark}

\begin {remark} \label {Rem11}
In the limit of $\hhbar$ small enough then $\frac {\nu}{f} = \frac {\eta C_1}{Fa} \sim C_1 \sim \hhbar^{-1/2}$ since Remark \ref {Rem7} and 
(\ref {F15}); therefore the stationary solutions ${\Vector {d}}^S$ takes the value $d_n^S =0$ if $n \notin S$ and $d_n^S \sim \pm 1$ if $n \in S$, and the energy 
$\mu^S$ belongs to an interval with center $\nu \sim C \hhbar^{3/2}$, for some $C>0$, and with amplitude of order $\hhbar^2$.
\end {remark}

\subsubsection {Bifurcation of stationary solutions}

We consider solution-sets $S$ associated to a given \emph {rung} of the (kind of) Stark-Wannier ladder satisfying the condition 
$S\subseteq [-N,+N]$ where $  \Chi(n)$ is a linear function. \ That is we consider energies $\mu^S$ in the interval $[\nu - f N , \nu + f N]$. \ We 
can see that stationary solutions to equation (\ref {F18}) associated to such solution-sets $S$ may bifurcate when the ratio $\nu/f$ is a positive integer number. 

In order to count how many stationary solutions we have let us introduce the following function (see Abramowitz and 
Stegun \cite {AS}, p. 825).

\begin {definition} \label {Definizione3} 
Let $Q(n)$, $n \in \N$, be the number of ways of writing the integer number $n$ 
as a sum of positive integers without regard to order, with the constraint that all integers in a given partition 
are distinct.
\end {definition}

E.g.: $Q(1)=1$, $Q(2)=1$, $Q(3) = 2$ and $Q(4)=2$. 

\begin {theorem} \label {Thm2}
When $\nu/f$ takes the value of a positive integer number then stationary solutions to (\ref {F18}), associated to solution-set 
$S \subset [-N,N]$, bifurcate. \ Furthermore, the total number of solutions-sets $S$ associated to a given \emph {rung} of the (kind of) 
Wannier-Stark ladder, assuming that all these sets $S$ are contained in the interval $[-N,+N]$,  is given by
\begin{eqnarray}
{M}(\nu /f) = \sum_{0< n < \nu /f} Q(n) \, . \label {F28}
\end{eqnarray}
\end {theorem}

\begin {proof}
First of all, because the stationary problem (\ref {F18}) is translation invariant $n \to n+\ell $ and $\mu^S \to \mu^S - f\ell$, 
provided that the solution-sets are contained in the interval $[-N,N]$, then we can always 
restrict ourselves to the $0$-th \emph {rung} of the ladder such that $ \min S=0$, that is the solution-set has the form 
$S= \{ 0, \ell_1 , \, \ldots \, , \ell_{{\mathcal N}-1} \}$ with $0<\ell_1 < \ell_2 < \ldots < 
\ell_{{\mathcal N}-1} < N$ positive and integer numbers. \ Hence, (\ref {F22}) becomes 
\begin{eqnarray*}
\mu^S = \frac {\nu}{\mathcal N} + \frac {f}{\mathcal N} \sum_{\ell \in S} \ell  \, . 
\end{eqnarray*}
and condition (\ref {F20}) implies the following condition on the solution-set $S$
\begin{eqnarray}
\frac {\nu}{f} > {\mathcal N} \max S - \sum_{\ell \in S} \ell = 
\sum_{\ell \in S} \left [  \max S - \ell \right ] > \sum _{\ell^\star \in S^\star} \ell^\star  
 \label {F29}
\end{eqnarray}
where 
\begin{eqnarray*}
S^\star = \{ \ell^\star := \max S - \ell \ : \ \ell \in S \} \, . 
\end{eqnarray*}
Let ${\mathcal S}^\star (\nu /f)$ be the collection of sets $S^\star$ satisfying (\ref {F29}), and  
let ${\mathcal Q}^\star (n)$ be the collection of sets of all non negative integer numbers, 
including the number $0$, which sum is equal to $n$, without regard to order with the 
constraint that all integers in a given partition are distinct; e.g. ${\mathcal Q}^\star (1) = \left \{ 
\{ 0,1\} \right \}$, ${\mathcal Q}^\star (2) = \left \{ \{ 0,2\} \right \}$ and ${\mathcal Q}^\star (3) = 
\left \{ \{ 0,3\} ,\, \{ 0,1,2\} \right \}$. \ Hence, by construction 
\begin{eqnarray*}
{\mathcal S}^\star (n+1) =  {\mathcal S}^\star (n) \cup {\mathcal Q}^\star (n) \, . 
\end{eqnarray*}
In conclusion, we have shown that the counting function ${M}(\nu /f)$ defined as the number of solution-sets $S$ of 
integer numbers satisfying the conditions (\ref {F29}) and such that $\min S =0$, is given by 
\begin{eqnarray*}
{M}(\nu /f) = \sum_{0< n < \nu /f} Q(n) \, . 
\end{eqnarray*}
Theorem \ref {Thm2} is so proved.
\end {proof}

\begin {remark}
A cascade of bifurcation points, when $\nu /f$ takes the value of any positive integer, occurs; indeed, when the 
ratio $\nu /f$ becomes larger than a positive integer $n$ then $Q(n)$ new stationary solutions appear. \ This 
fact can be seen in Figure \ref {Figura2}, where we plot the values of ${\mu}/f$, when $\nu /f$ belongs 
to the interval $[0,10]$, associated to solution-sets $S$ such that $\min S =0$, that is we plot the value of 
energies associated to the $0$-th \emph {rung} of the (kind of) Wannier-Stark ladder. \ By translation $\mu 
\to \mu + j f $, $j \in {\mathbb Z}$, and thus this picture occurs for each \emph {rung} of the 
ladder and then the collection of values of $\mu$ associated to stationary solutions is going to densely cover intervals of the real axis.
\begin{center}
\begin{figure}
\includegraphics[height=7cm,width=8cm]{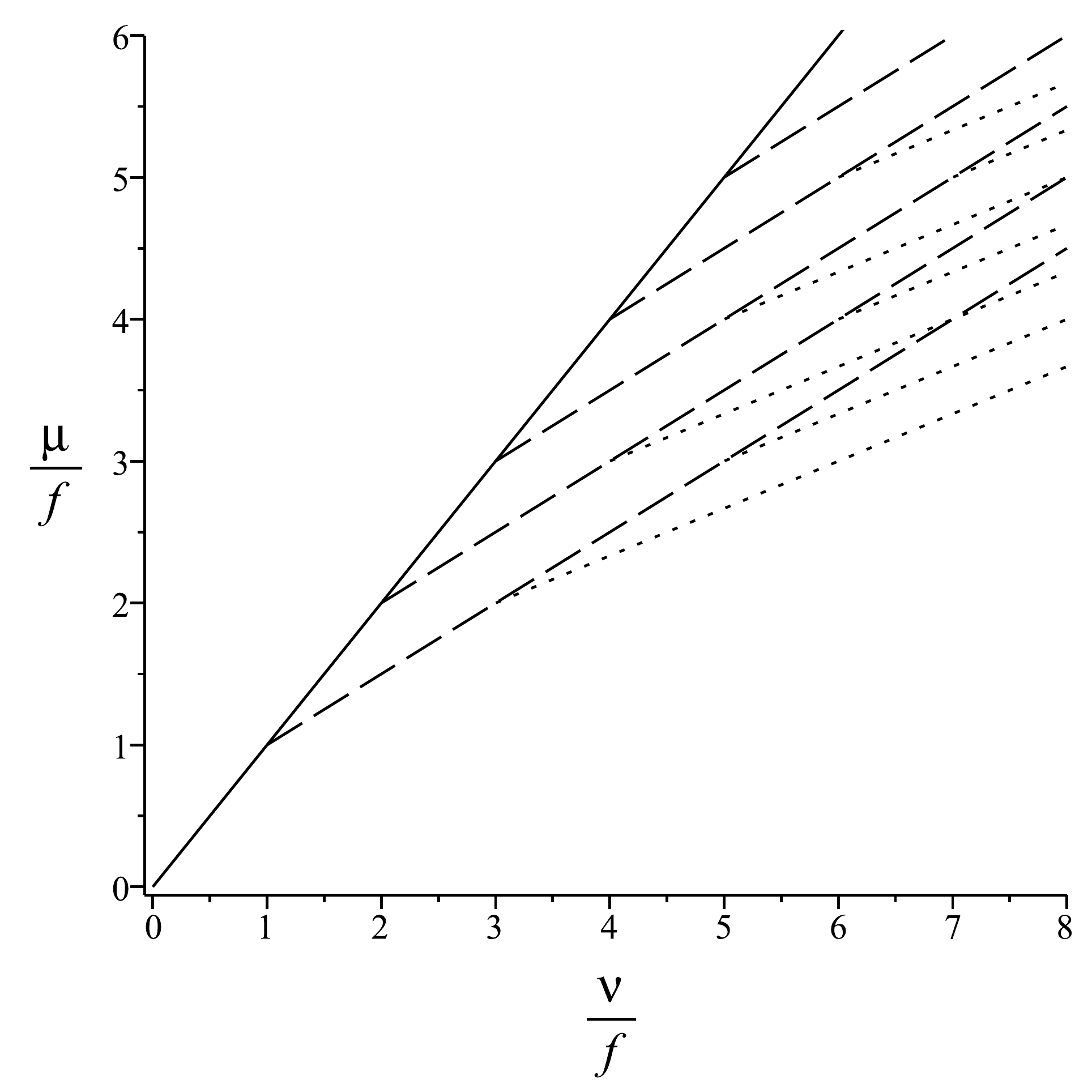}
\caption{\label {Figura2} Here we plot the values of $\mu/f$ associated to stationary solution-sets $S$ such that $\min S =0$ and where ${\mathcal N}=1,2,3$; 
we can see a cascade of bifurcations when $\nu/f$ increases. \ Full line represents the solution corresponding to the $0$-th \emph {rung} 
of the Stark-Wannier ladder localized on the $0$-th cell (${\mathcal N}=1$), broken lines represent the solutions of the same \emph {rung} of the 
Stark-Wannier ladder localized on two cells (${\mathcal N}=2$), and finally point lines represent the solutions of the same \emph {rung} of the 
Stark-Wannier ladder localized on three cells (${\mathcal N}=3$).}
\end{figure}
\end{center}
\end {remark}

\begin {remark}
One can see that ${M}(\nu /f)$ grows quite fast, indeed the following asymptotic behavior holds true \cite {AS}:
\begin{eqnarray*}
 Q(n)\sim \frac {e^{\pi \sqrt {n/3}}}{4\cdot 3^{1/4}n^{3/4}} \ \mbox { as } \ n \to \infty \, .
\end{eqnarray*} 
Hence
\begin{eqnarray*}
{M}(n) \sim \frac {1}{2} \mbox {erfi} \left [ \sqrt {\pi} (n/3)^{1/4} \right ] \sim \frac {\exp \left [  
{\pi} (n/3)^{1/2} \right ] }{2 {\pi} (n/3)^{1/4}} 
\end{eqnarray*}
as $n$ goes to infinity, where $\mbox {erfi} (x) = - i\, \mbox {erf} (ix)$ is the imaginary error function. \ In particular, because 
(see Remark \ref {Rem11}) $\frac {\nu}{f}  \sim C_1 \sim \hhbar^{-1/2}$, then we have that the energy $\mu$ lies 
in an interval $[\nu - f N , \nu + f N]$ with center at $\nu \sim \hhbar^{3/2}$ and amplitude of order $\hhbar^2$, and the number of stationary solutions is of order 
\bee
M \left ( \frac {\nu}{f} \right ) \sim \hhbar^{1/8} e^{C \hhbar^{-1/4}}\, \ \mbox { as } \ \hhbar \ \mbox { goes to zero,} 
\eee
for some positive constant $C$. \ That is the energy spectrum densely fill the interval $[\nu - f N , \nu + f N]$  when $\hhbar $ goes to zero.
\end {remark}

\begin {remark}
Since we assumed that the parameters $\hbar$, $m$, $\alpha_1$ and $\alpha_2$ are fixed and that $\epsilon \ll 1$ then the rescaling (\ref {res}) 
immediately implies that the effective nonlinear coupling strength parameter $\eta$ is of order $\hhbar^2$, where $\hhbar$ plays the role of a semiclassical 
parameter. \ Hence, the bifurcation parameter is such that 
\bee
\frac {\nu}{f} = \frac {\eta C_1}{F a} \sim \hhbar^{-1 /2} \gg 1 
\eee
and we have a dense energy spectrum. \ One can also consider the case where the parameters depend on $\epsilon$; a similar approach has been used, e.g., 
by \cite {FS1}. \ If the nonlinear coupling strength parameter $\alpha_2$ depends on some power by $\epsilon$ then $\eta \sim \hhbar^{u}$ for some 
power $u$ and 
\bee
\frac {\nu}{f} = \frac {\eta C_1}{F a} \sim \hhbar^{u-5/2} \, .  
\eee
If $u < \frac 52 $ then we have a dense energy spectrum as in the case above; if $u > \frac 52$ then we have a single Stark-Wannier 
ladder; if $u = \frac 52$ then we observe the bifurcation phenomenon. 
\end {remark}

\subsubsection {When do ${\mathcal N}$-mode stationary solutions arise from $({\mathcal N}-1)$-mode stationary solutions?} If one looks with more detail the 
bifurcation cascade one can see that we have ${\mathcal N}$-mode solutions for any value of ${\mathcal N}$, provided that $S\subset [-N,N]$ for some $N$ 
large enough. \ Let us restrict our analysis, for sake of simplicity, to solution-sets $S$ contained in the interval $[-N,N]$ where $  \Chi (n)$ is a linear function. 

As said above, ${\mathcal N}$-mode stationary solutions are associated to solution-sets of the form
\begin{eqnarray}
S = \{  j , \  j+\ell_1,\ \ldots \ ,\  j+\ell_{{\mathcal N}-1} \} \label {F31}
\end{eqnarray}
under condition (\ref {F20}). \ Now, let us consider, as a particular family of ${\mathcal N}$-mode solutions, solution-sets of the form (\ref {F31}) for any 
$ j \in {\mathbb Z}$ and $\ell_{r+1}-\ell_{r} =1$, assuming that $|j|, |j+ {\mathcal N} -1| \le N$. \ They are associated to
\begin{eqnarray*}
\mu^S = \frac {\nu}{{\mathcal N}} + f  j + \frac 12 f ({\mathcal N}-1)
\end{eqnarray*}
and then condition (\ref {F20}) implies that 
\begin{eqnarray*}
\frac {{\mathcal N}({\mathcal N}-1)}{2} <  \frac {\nu}{f}  
\end{eqnarray*}
Hence, we can observe a second bifurcation phenomenon: stationary solutions associated to solution-sets 
with ${\mathcal N}$ elements arise from stationary solutions associated to solution-sets with ${\mathcal N}-1$ elements when $\nu /f$ 
becomes bigger than the critical value $\frac 12 {{\mathcal N}({\mathcal N}-1)}$.

We can summarize such a result as follows

\begin {theorem} \label {Thm3}
If $\nu/f <  {{\mathcal N}({\mathcal N}-1)}/2 $ then stationary solutions 
to (\ref {F17}), associated to solution-sets $S \subset [-N , +N]$, are localized on a number of sites less than ${\mathcal N}$, 
at $\nu/f =  {{\mathcal N}({\mathcal N}-1)}/2 $ a stationary solution localized on ${\mathcal N} -1$ sites bifurcates and a new stationary 
solution localized on ${\mathcal N}$ sites arises.
\end {theorem}

\section {Existence of solutions to the DNLSWE} \label {Sec5}

Let $\mu^S$ and $\Vector {d}$ b a finite-mode solution to (\ref {F17}) associated to a solution-set $S$ given by Theorem \ref {Thm1}. \ Now, we will 
prove, by a stability argument, that this solution becomes a solution to 
(\ref {F16}) when $\beta$ is small enough. \ To this end we have to remind that $\beta$ goes to zero when $\hhbar$ goes to zero according with Lemma \ref {Lemma8}.i.

\begin {theorem} \label {Thm4}
Let $\frac {\nu}{f} \notin \N$. \ Let $S$ be a solution-set to (\ref {F17}) with associated energy $\mu^S$ and normalized 
stationary solution $\Vector {d}^S$  given by (\ref {F19}). \ We assume that $S \subset [-N,N]$. \ Then, if $\beta$ is small enough there exists a stationary 
solution $\Vector {g}^S \in \ell^1_{\R}$ to the DNLSWE (\ref {F16}) associated to $\tilde \lambda = \mu^S$ and such that 
\bee
\| \Vector {g}^S - \Vector {d}^S \|_{\ell^1} = \aasy \left ( e^{- S_0/\hhbar } \right ) \ \mbox { and } \ \| \Vector {g}^S \|_{\ell^2} =1 + \aasy 
\left ( e^{-S_0/\hhbar } \right ) \, . 
\eee
\end {theorem}

\begin {proof} First of all let us recall that from Remark \ref {Rem11} then $\mu^S \sim C \hhbar^{3/2} \not= 0$. \ In the following let us omit the 
upper letters $S$ in $\Vector {d}^S$ and $\Vector {g}^S$ for sake of simplicity. \ If we rescale $g_n = 
\left [ \frac {\tilde \lambda}{\nu} \right ]^{1/2} {g'}_n $ and $d_n = \left [ \frac {\tilde \lambda}{\nu} \right ]^{1/2} {d'}_n 
$, and if we set $\beta'= \beta /{\tilde \lambda}$ and $f'= f/{\tilde \lambda}$ then equations (\ref {F16}) and (\ref {F17}) 
take the form 
\be
\left ( 1 - {g'}_n^{2} \right ) {g'}_n = \beta' ( {g'}_{n+1} + {g'}_{n-1}  )+ f'   \Chi (n) {g'}_n \, . \label {F32}
\ee
with anticontinuous limit
\be
\left ( 1 - {d'}_n^{2} \right ) {d'}_n =  f'   \Chi (n) {d'}_n \, . \label {F33}
\ee
Therefore, any solution ${\tilde \lambda}$ and $\Vector {g}$ to (\ref {F16}) is associated to a solution $\Vector {g}'$ 
to (\ref {F32}), and all the solutions $\Vector {d}'$ to (\ref {F33}) associated to the values $\tilde \lambda = \mu^S$ given by Theorem \ref {Thm1} are isolated in 
$\ell^1$ by construction when there are no bifurcations, that is for $\frac {\nu}{f} \notin \N$.

Let ${\mathcal F}_1 : (-\delta , + \delta ) \times \ell^1_{\R} (\Z ) \to \ell^1_{\R} (\Z )$ be the map defined as 
\bee
\left ( {\mathcal F}_1 (\beta ' , {\Vector {g}'} ) \right )_n := - \left ( 1 - {g'}_n^{2} \right ) {g'}_n + 
\beta' ( {g'}_{n+1} + {g'}_{n-1}  )+ f'   \Chi (n) {g'}_n \, . 
\eee
We are going to look for solutions ${\Vector g}' (\beta ')$ to equation ${\mathcal F}_1 (\beta ' , {\Vector {g}'} ) =0$; where we 
already know that equation ${\mathcal F}_1 (0 , {\Vector {d}'} ) =0$ has solutions ${\Vector d}' $ associated to the ones given by Theorem \ref {Thm1}. \ We may extend 
the solutions to (\ref {F33}), obtained in the anticontinuous limit $\beta ' \to 0$, to the solutions to equation (\ref {F32}) 
for $\beta'$ small enough if the tridiagonal matrix 
\bee
T (\beta ') := \left ( D_{{\Vector g}' } {\mathcal F}_1 \right ) (\beta ', {\Vector g}' )= \mbox {tridiag} ( \beta' ,  f'   \Chi (n) -1 + 
3 {g'}_n^{2}  , \beta' ) \, , 
\eee
obtained deriving the previous equation by ${g'}_n$, is not singular at $\beta' =0$, where ${g'}_n$ takes the value of a solution ${d'}_n$ to (\ref {F23}) 
obtained for $\beta' =0$. \ The linearized map $\left ( D_{{\Vector g}' } {\mathcal F}_1 \right ) (0, {\Vector d}' )$ can be written as $\mbox {diag} (T_n)$, where 
\be
T_n =   \frac {f   \Chi (n) }{\mu^S} -1 + 3 {d'}_n^{2} \ \mbox { and } \ {d'}_n^{2} = 
\left \{
\begin {array}{ll}
 0 & \ \mbox { if } n \notin S \\ 
 \frac {\mu^S - f   \Chi (n) }{\mu^S} & \ \mbox { if } n \in S
\end {array}
\right. \, . \label {F34}
\ee
Hence
\bee
T_n = 
\left \{
\begin {array}{ll}
\frac {f  \Chi (n)-\mu^S}{\mu^S } & \ \mbox { if } n \notin S \\ 
 -2 \frac {f  \Chi (n)-\mu^S}{\mu^S } 
& \ \mbox { if } n \in S
\end {array}
\right. \, . 
\eee

\begin {lemma} \label {Lemma4}
Let $\hhbar$ be small enough, then it follows that 
\bee
\inf_{n\in \Z} |T_n| > \frac 12  \, .
\eee
\end {lemma}

\begin {proof} Assume at first that ${\mathcal N}=1$. \ In this case $S=\{ j \}$ for any $j \in \Z$ such that $|j|\le N$ and 
\bee
\mu^S =\nu + f  \Chi (j) = \nu + f j + f \frac {C_0}{a} + \aasy \left ( e^{-S_0/\hhbar } \right ) 
\eee
as $\hhbar$ goes to zero, because of Lemma \ref {Lemma2}. \ Hence
\bee
T_n = \frac {f [  \Chi (n) -  \Chi (j)] - \nu}{\mu^S} + \aasy \left ( e^{-S_0/\hhbar } \right )  \, .
\eee
Recalling now that $\nu = \eta C_1 \sim \hhbar^{3/2}$ and $f =Fa \sim \hhbar^2$ then $T_n \sim -1$. \ Similarly, 
we can easily extend such arguments to any integer number ${\mathcal N}>1$. \  In this case $S= \{ j , j+\ell_1 , \ldots , j +\ell_{{\mathcal N}-1} \}$, 
for $j \in \Z$ and $\ell_k \in \N$ such that $|j|$, $|j+\ell_k |<M$, $k =1,\ldots , {\mathcal N}-1$. \ Then 
\bee
\mu^{S} =\frac 1{\mathcal N} \nu + \frac f{\mathcal N} \sum_{\ell \in S} \ell + f \frac {C_0}{a} \sim \hhbar^{3/2} \, ; 
\eee
hence
\bee
T_n = \frac {f ({\mathcal N}   \Chi (n)- \sum_{\ell \in S} \ell - \frac {{\mathcal N}C_0}{a} ) - \nu}{{\mathcal N} \mu^S}  \sim -1 \, , \forall n  \, . 
\eee
The proof of the Lemma is so completed.

\end {proof}

Now we are ready to conclude the proof of Theorem \ref {Thm4}. \ Indeed, by Lemma \ref {Lemma4}, the linearized map 
$\left ( D_{\Vector {g}'} {\mathcal F}_1 \right )\left (0, \Vector {d}' \right )$ is invertible with inverse uniformly 
bounded. \ Therefore, by the Implicit Function Theorem, there exists a neighborhood ${\mathcal U}$ of $0$ such that if $\beta'$ 
belongs to such a neighborhood ${\mathcal U}$ then there exists a unique solution $\Vector {g}' := \Vector {g}' (\beta ')$ to equation 
${\mathcal F}_1 \left (\beta' , \Vector {g}' \right )=0$ in a $\ell^1$-neighborhood of $\Vector {d}' $, where $\Vector {d}' $ 
is an isolated solution to ${\mathcal F}_1 \left (0, \Vector {d}' \right )=0$ because bifurcations occur at $\frac {\nu}{f} \in \N$. \ Since $\beta ' = \aasy \left (e^{-S_0/\hhbar} \right )$, 
for any $\hhbar $ small enough, then $\beta' \in {\mathcal U}$ for any $\hhbar \in (0, \hhbar^\star )$ for some $\hhbar^\star >0$. \ From this fact and because 
the map $\beta' \to \Vector {g}' (\beta ')$ is $C^1$ then we can conclude that when $\hhbar$ is small enough then there exists a 
solution $\Vector {g}' \in \ell^1_\R$ to equation ${\mathcal F}_1 \left (\beta' , \Vector {g}' \right )=0$ such that 
\be
\| \Vector {g}' - \Vector {d}'  \|_{\ell^1}= \aasy \left (e^{-S_0/\hhbar} \right ) \, . \label {F35}
\ee
By construction it follows also that when $\tilde \lambda = \mu^S$ then 
\bee
\| \Vector {d} \|_{\ell^2} =1\, , \ \| \Vector {d'} \|_{\ell^2} =\left ( \frac {\nu}{\mu^S} \right )^{1/2} \ \mbox { and } \ \| \Vector {g}' \|_{\ell^2} 
= \left ( \frac {\nu}{\mu^S} \right )^{1/2 } \| \Vector {g} \|_{\ell^2}
\eee
hence
\bee
\left | \, \| \Vector {g} \|_{\ell^2} -1 \right | \le \| \Vector {g} - \Vector {d} \|_{\ell^2} \le \| \Vector {g} - \Vector {d} \|_{\ell^1} 
= \aasy \left (e^{-S_0/\hhbar} \right ) 
\eee 
from (\ref {F35}). \ Then the proof of the Theorem is given. 
\end {proof}

\begin {remark} \label {Rem14}
Since we can always normalize $\Vector {g}^S $ to $1$ by means of suitable rescaling of the nonlinearity parameter $\nu$ we can conclude 
that $\Vector {g}^S $ is a normalized solution to (\ref {F16}) associated to $\mu^S$ for some $ \tilde \nu = \nu + \aasy \left ( 
e^{-S_0/\hhbar } \right )$. \ Furthermore, by construction (see \S 3.2), the map $\nu \to \mu^S (\nu )$ is $C^1$ when 
we are far form the bifurcation points $\frac {\nu}{f} \in \N$; then we can conclude that for any $\nu$ fixed and such that $\frac {\nu}{f} \notin \N$ then 
equation (\ref {F16}) has a solution $\tilde \lambda$ and $\Vector {g}^S$ where $\Vector {g}^S$ is normalized and it satisfies (\ref {F35}) and $\tilde \lambda = \mu^S + \aasy 
\left ( e^{-S_0/\hhbar } \right )$. 
\end {remark}

\section {Fixed point argument} \label {Sec6}

Here, we go back to equation (\ref {F9}) and, at first, we justify the existence of $\psi_{\perp}$ by means of a fixed point argument. \ Recalling that 
$\lambda =\tilde  \lambda + (\Lambda_1 + F C_0 )$ and $\tilde \lambda = \mu^S + \aasy (e^{-S_0/\hhbar } )$ where $\mu^S \in [\nu - f N , \nu + f N ]$ 
and where $|C_0| \le C$, then the value of $\lambda$ corresponding to $\mu^S$ is such that $\lambda = \Lambda_1 + \nu + \asy ( \hhbar^2 )$. \ Hence, we consider the 
second equation of (\ref {F9}) for $\lambda$ in a neighborhood of $\Lambda_1$ with width of order $\hhbar^{3/2}$.

\begin{theorem} \label{Thm5} 
Let $\frac {\nu}{f} \notin \N$, $\psi = \psi_1 + \psi_\perp \in L^2$, where $\psi_1 = \Pi \psi = \sum_{n\in \Z} c_n \w_n$ and 
$\psi_\perp =\Pi_\perp \psi$, let $\hhbar >0$ small enough. \ Let $\delta_0 >0$ be any fixed real and positive number, then for any 
$\Vector{c}=(c_n)_{n\in \Z} \in \ell^1_{\R} (\Z) $, with $\|\Vector{c}\|_{\ell^1 (\Z)} \le \delta_0$, there exists a unique smooth map 
\bee
\hat{\psi}_\perp : \ell^1_{\R} (\Z) \to H^1(\R)
\eee
such that $\psi_{\perp}=\hat{\psi}_\perp (\Vector{c})$ is a solution to the second equation of (\ref {F9}) for small $\hhbar>0$. \ Moreover, 
$\hat{\psi}_\perp (\Vector{c})$ is small as $\hhbar \to 0$ in the sense that there exists a positive constant $C>0$, dependent on $\delta_0$ and independent 
of $\hhbar$, such that 
\be
\|\hat {\psi}_{\perp} (\Vector {c}) \|_{H^1} \le C  \hhbar^{ {1}/{2} } \, .  \label {F36}
\ee
\end{theorem}

\begin {proof} We make use here of same ideas already developed by \cite {FS2} adapted to the case of a tilted periodic potential. \ Let $\Lambda_1$ be defined as in 
\S \ref {A.2} and let $E\in \R$ be fixed. \ Note that the operator $H_B + F \W-(\Lambda_1- \Gamma E)$ on $\Pi_{\perp}L^2$ has inverse operator for $\hhbar$ sufficiently small 
provided that 
\bee
\Gamma = \asy (\hhbar^{3/2} )
\eee
and thanks to the fact that the $dist [\Lambda_1 , \sigma (\Pi_\perp H_B ) ] = O(\hhbar)$, that $\W$ is a bounded operator, and 
that $F= \asy (\hhbar^2 )$ from (\ref {F3}). \ Precisely, there exists a constant $C>0$ independent of $\hhbar$ such that  
\bee
\| [ H_B +F \W -(\Lambda_1-\Gamma E) ]^{-1} \Pi_{\perp}\|_{\mathcal{L}(L^2 \to H^1)} \le C \hhbar^{-1}\, . 
\eee
Then the second equation of (\ref {F9}) may be written as
\be
\psi_{\perp}={\mathcal F}_2(\psi_{\perp}), \label {F37}
\ee
where $\Vector {c} \in \ell^1_{\R} $ is fixed and where we set $\Lambda_1 = \lambda + \Gamma E$, $\psi=\psi_1 +\psi_{\perp}$ and 
\be
{\mathcal F}_2 (\psi_{\perp}) = \Pi_\perp \left [ H_B +F W -(\Lambda_1-\Gamma E)\right ]^{-1} \Pi_\perp \left \{ - \eta 
\psi^{3} - F \W \psi_1 \right \} \, . \label {F38}
\ee
We are going to show that ${\mathcal F}_2$ is a contraction map in 
\bee
K_\gamma = \{\psi_\perp \in H^1(\R) \cap \Pi_\perp L^2(\R) : \|\psi_\perp \|_{H^1} \le \gamma \}
\eee
for some $\gamma >0$. \ Indeed, let $\psi_\perp$, $\phi_\perp \in K_\gamma$ and let $\psi=\psi_1 +\psi_{\perp}$ and 
$\phi=\psi_1 +\phi_{\perp}$, we have 
\bee
\| {\mathcal F}_2 (\psi_{\perp}) \|_{H^1} 
&\le& C \left [ \frac{|\eta |}{\hhbar}  \| \psi^{3} \|_{L^2} + \frac {F}{\hhbar}  \| \W \psi_1 \|_{L^2} \right ] \\
&\le& \frac{|\eta | C}{\hhbar} \left ( \|\psi_1 \|_{L^{6}}^{3} +
\|\psi_{\perp}\|_{L^{6}}^{3} \right ) + C \hhbar \| \Vector {c} \|_{\ell^1}  
\eee
since Remark \ref {Rem3} and (\ref {F3}). \ Then, by the Gagliardo-Nirenberg inequality, it follows that 
\bee
\| \psi_\perp \|_{L^{6}}^{3} \le C \| \partial_x \psi_\perp \|_{L^2} \| \psi_\perp 
\|_{L^2}^{2} \le C \gamma^{3} 
\eee
since $\psi_\perp \in K_\gamma$, and because by Remark \ref {Rem3} and 
Lemma \ref {Lemma8}.vi 
\bee
\| \psi_1 \|_{L^{6}}^{3} \le C \hhbar^{-1/2} \| \Vector {c} \|_{\ell^1}^{3} 
\le C \hhbar^{-1/2} \delta_0^{3}
\eee
Hence
\bee
\| {\mathcal F}_2 (\psi_{\perp}) \|_{H^1} \le C_2 := C_2 (\hhbar ) = C \left [ |\eta | \hhbar^{-3/2} + |\eta | \hhbar^{-1} 
\gamma^{3}+ \hhbar \right ] = C \hhbar^{1/2} 
< \gamma 
\eee
for some $\hhbar $ small enough. \ Furthermore
\bee
{\mathcal F}_2(\psi_{\perp})-{\mathcal F}_2(\phi_{\perp}) = - \eta \Pi_\perp \left [ H_B +F W -(\Lambda_1-\Gamma E)\right ]^{-1} \Pi_\perp 
\left ( \psi^{3} -  \phi^{3} \right )  \, , 
\eee
hence
\bee
&& \|{\mathcal F}_2(\psi_{\perp})-{\mathcal F}_2(\phi_{\perp})\|_{H^1} 
\le \frac{|\eta | C }{\hhbar}  \left ( \|\psi\|_{L^{4}}^{2} + \| \phi \|_{L^{4}}^{2}\right ) 
\|\psi_{\perp} - \phi_{\perp} \|_{H^1} \\
&& \ \ \le  C |\eta | \hhbar^{-1} \left ( \| \psi_\perp \|_{L^{4}}^{2} + \| \phi_\perp \|_{L^{4}}^{2} + 
\| \psi_1 \|_{L^{4}}^{2} \right ) \|\psi_{\perp} - \phi_{\perp} \|_{H^1}\\ 
&& \ \ \le  C |\eta |\hhbar^{-1} \left ( \| \partial_x \psi_\perp \|_{L^2}^{{1}/{2}} \| \psi_\perp 
\|_{L^2}^{{3}/{2}}+ \| \partial_x \phi_\perp \|_{L^2}^{{1}/{2}} \| \phi_\perp \|_{L^2}^{{3}/{2}}+
\hhbar^{- {1}/{4}} \delta_0^{2} \right ) \|\psi_{\perp} - \phi_{\perp} \|_{H^1}\\
&& \ \ \le  C |\eta | \hhbar^{-1} \left ( \gamma^{{1}/{2}}+\hhbar^{-{1}/{4}} \delta_0^{2} \right ) 
\|\psi_{\perp} - \phi_{\perp} \|_{H^1}\\
&& \ \ \le C_3 \|\psi_{\perp} - \phi_{\perp} \|_{H^1}    
\eee
where 
\bee
C_3 := C_3 (\hhbar ) = C |\eta | \hhbar^{- {5}/{4}}  <1
\eee
since 
\bee
\eta = \asy (\hhbar^2) \, . 
\eee
Then there exists a unique solution $\hat {\psi}_{\perp} = \hat {\psi }_{\perp} (\Vector {c}  ) \in K_\gamma$ to equation (\ref {F37}) 
for small $\hhbar>0$. \ Moreover, by construction the solution $\hat \psi_\perp$ is given by 
\bee
\hat \psi_\perp = \sum_{j=1}^\infty (\psi_{\perp ,j} - \psi_{\perp ,{j-1}} ) + \psi_{\perp ,0}
\eee
where $\psi_{\perp ,j} = {\mathcal F}_2 (\psi_{\perp ,{j-1}})$ and $\psi_{\perp ,0} =0$. \ Hence 
\bee
\| \psi_{\perp ,j} - \psi_{\perp ,{j-1}} \|_{H^1} &=& \| {\mathcal F}_2 (\psi_{\perp ,{j-1}}) - {\mathcal F}_2 (\psi_{\perp ,{j-2}}) \|_{H^1} \\ 
& \le & C_3 \| \psi_{\perp ,{j-1}} - \psi_{\perp ,{j-2}} \|_{H^1} \\ 
& \le & C_3^{j-1} \| \psi_{\perp ,{1}} - \psi_{\perp ,{0}} \|_{H^1} \le C_2 C_3^{j-1} 
\eee
and thus
\bee
\| \hat \psi_\perp \|_{H^1} \le \sum_{j=1}^\infty C_2 C_3^{j-1} = 
\frac {C_2 }{1-C_3} \le C \hhbar^{1/2 }
 \eee
for some positive constant $C$. \ This fact completes the proof.
\end {proof}

We must underline that $\psi_1$ linearly depends on $\Vector {c}$ and thus the map $\Vector {c} \to \hat \psi_\perp (\Vector {c} )$ is a smooth 
map. \ In particular the following result holds true.

\begin {lemma} \label {Lemma5}
Let $\Vector {c} \in \ell^1_\R$ be such that $\| \Vector {c} \|_{\ell^1} \le \delta_0$, where $\delta_0$ is any fixed and positive 
number. \ Then for any $\Vector {q}$ such that $ \| \Vector {c} + \Vector {q} \|_{\ell_1} \le \delta_0$ then 
$\hat \psi_\perp ( \Vector {c} + \Vector {q} )$ there exists and it is such that
\bee
\| \hat \psi_\perp ( \Vector {c} + \Vector {q} ) - \hat \psi_\perp ( \Vector {c} ) \|_{H^1} \le C \hhbar \| \Vector {q} \|_{\ell^1} \, . 
\eee
\end {lemma}

\begin {proof}
Indeed, equation (\ref {F37}) becomes 
\be
\psi_\perp + \psi_q = {\mathcal F}_2 \left ( \psi_\perp + \psi_q \right ) \label {F39}
\ee
where we set
\bee
\psi_\perp := \hat \psi_\perp (\Vector {c} ) \ \mbox { and } \ \psi_q := \hat \psi_\perp (\Vector {c} + \Vector {q})- \hat \psi_\perp (\Vector {c} ) \, . 
\eee
A straightforward computation gives that
\bee
{\mathcal F}_2 \left ( \psi_\perp + \psi_q \right ) = {\mathcal F}_2 \left ( \psi_\perp \right ) + {\mathcal R}
\eee
where
\bee
{\mathcal R} &=&  \Pi_\perp \left [ H_B +F W -(\Lambda_1-\Gamma E)\right ]^{-1} \Pi_\perp \times \\ 
&& \ \times \left \{ - \eta 
\left ( 3 \psi_\perp^2 \psi_q + 3 \psi_\perp \psi_q^2 + \psi_q^3 \right ) - F W \left ( \sum_{n\in \Z} q_n \w_n \right )  \right \} 
\eee
and (\ref {F39}) reduces to 
\be
\psi_q = {\mathcal R} \, . \label {F40}
\ee
The same arguments used in the proof of Theorem \ref {Thm5} yields to the following estimate
\bee
\| \psi_q \|_{H^1} &=& \| {\mathcal R}  \|_{H^1}  \\ 
&\le & C \left [ \frac {|\eta |}{\hhbar} \left ( \| \psi_\perp \|^2_{L^\infty}  + 
\| \psi_\perp \|_{L^\infty} \| \psi_q \|_{L^\infty} + \| \psi_q \|^2_{L^\infty} \right ) \| \psi_q \|_{L^2} + \frac {|F|}{\hhbar} \| \Vector {q} \|_{\ell^1} \right ] \\ 
& \le & C \hhbar^2 \| \psi_q \|_{H^1} + C \hhbar  \| \Vector {q} \|_{\ell^1}
\eee
because $\| \psi_\perp \|_{H^1}, \, \| \psi_q \|_{H^1} \le C \hhbar^{1/2}$ and (\ref {F3}). \ Then $\| \psi_q \|_{H^1} \le C \hhbar  \| \Vector {q} \|_{\ell^1}$ immediately follows.
\end {proof}

\begin {remark} \label {Rem15}
From Lemma \ref {Lemma5} it follows that the linear map $D_{\Vector {c}} (\hat \psi_\perp )  $ 
satisfies the estimate 
\bee
\left \| D_{\Vector {c}} (\hat \psi_\perp ) \right \|_{{\mathcal L}(\ell^1 \to H^1 )} \le C \hhbar \, .
\eee
\end {remark}

\section {Existence of stationary solutions} \label {Sec7}

\begin {theorem} \label {Thm6} Let $\frac {\nu}{f} \notin \N$ and let $\hhbar >0$ small enough. \ Let 
$\Vector {d}^S$ be a finite-mode normalized solution associated to a solution-set $S$ satisfying the assumption of Theorem \ref {Thm4}. \ 
Then there exists a stationary solution $\psi^S$ to equation (\ref {F12}) such that 
\bee
\left \| \psi^S - \sum_{n\in S} d_n^S \w_n \right \|_{H^1} \le C \hhbar^{1/4}  \, .  
\eee
\end {theorem}

\begin {proof}

Let us omit, for the sake of simplicity, the upper letter $S$. \ We have to consider the first equation of (\ref {F12}) where $\tilde \lambda$, $f$ and $\nu$ are 
defined by (\ref {F13Bis}):
\be
\tilde \lambda c_n = - \beta (c_{n+1} + c_{n-1}  ) + f   \Chi (n) c_n + \nu c_n^{3} + r^n \label {F41}
\ee
where $r_n$ is defined by (\ref {F13}) and where the map
\bee
 ( \Vector {c} , \psi_\perp ) \in \ell^1_{\R} \times H^1 \to \Vector {r} = \{ r^n \}_{n\in \Z} \in \ell^1_{\R} (\Z)
\eee
is norm bounded by (see Lemma \ref {Lemma1}, Lemma \ref {Lemma2}, equation (\ref {F15}), Lemma \ref {Lemma3} and Theorem \ref {Thm5})
\bee
\| \Vector {r} \|_{\ell^1} & \le & \| \Vector {r}_1 \|_{\ell^1} + |F| \, \left [ \| \Vector {r}_2 \|_{\ell^1} + 
\| \Vector {r}_3 \|_{\ell^1} \right ] + |\eta |\, \| \Vector {r}_4 \|_{\ell^1} + |F| \| \Vector {r}_5 \|_{\ell^1} \\ 
& \le & C e^{-(S_0 + \zeta )/\hhbar } \| \Vector {c} \|_{\ell^1} + C_\rho e^{-(S_0 - \rho )/\hhbar } \| \Vector {c} \|_{\ell^1} + 
C |F| \hhbar^{1/2}  + \\ 
&& \ + C_\rho \| \Vector {c} \|_{\ell^1}^{3} e^{-(S_0 - \rho )/\hhbar } + C |\eta | \hhbar^{1/4} \| \Vector {c} \|_{\ell^1}^2 + 
 C_\rho e^{-(S_0 - \rho )/\hhbar } \| \Vector {c} \|_{\ell^1}
\eee
for some $\zeta >0$ and for any $\rho \in (0, S_0)$, where $C$ is a positive constant and $C_\rho$ is a positive constant depending on $\rho$. 

\ Now, let us consider the 
following mapping
\bee
( \Vector {c} , \kappa ) \in \ell^1_{\R} (\Z ) \times \R \to {\mathcal G} ( \Vector {c} , \kappa ) = \left \{ {\mathcal G}_n ( \Vector {c} , \kappa ) \right \}_{n\in \Z} 
\in \ell^1_{\R } (\Z ) 
\eee
defined as 
\be
{\mathcal G}_n ( \Vector {c} , \kappa ) = \tilde \lambda c_n + \beta (c_{n+1} + c_{n-1}  ) - f   \Chi (n) c_n - \nu c_n^{3} - \kappa \hhbar^{-2} 
r^n \label {F42}
\ee
where $\Vector {r} := \Vector {r} (\Vector {c}) = \Vector {r} ( \Vector {c}, \psi_\perp )$ and where $\psi_\perp = \hat \psi_\perp \left (\Vector {c} \right )$ is the solution to 
the second equation of (\ref {F12}) for small $\hhbar >0$ given by Theorem \ref {Thm5}.

By construction, ${\mathcal G}_n ( \Vector {c} , 0 ) =0$ coincides with the discrete nonlinear Schr\"odinger equation DNLSWE (\ref {F14}), while 
${\mathcal G}_n ( \Vector {c} , \hhbar^2 ) =0$ coincides with equation (\ref {F41}).

\begin {lemma} \label {Lemma6}
$\mathcal{G}$ is a $C^1$ map in $(\Vector{c}, \kappa)$. \ In particular: 
 
\begin {itemize}

\item [i.] for any fixed $\rho \in (0, S_0)$ there exists 
 a positive constant $C:=C_{\rho } >0$ such that: he map $\Vector {r}_1 : \ell^1_\R \to \ell^1_\R$ satisfies
\be 
\|\left ( D_{\Vector{c}} \Vector{r}_1 \right ) (\Vector{c})\|_{{\mathcal L}(\ell^1 \to \ell^1)}
\le C (1+\|\Vector{c}\|_{\ell^1}^2 ) e^{-(S_0-\rho)/\hhbar}. \label {F43}
\ee

\item [ii.] the maps $\Vector {r}_2 : \ell^1_\R \to \ell^1_\R$ and $\Vector {r}_5 : \ell^1_\R \to \ell^1_\R$ are 
linear maps such that 
\be 
\|\left ( D_{\Vector{c}} \Vector{r}_j \right ) (\Vector{c})\|_{{\mathcal L}(\ell^1 \to \ell^1)}
= \aasy \left ( e^{-S_0/\hhbar} \right ) \, ,\ j=2,5 . \label {F44}
\ee

\item [iii.] the map $\Vector {r}_3 : H_1 \to \ell^1_\R$ does not directly depend on $\Vector {c}$ and it is such that 
\be 
\left \| \left ( D_{\Vector{c}} \Vector{r}_3 \right ) \left ( \hat \psi_\perp (\Vector {c} ) \right ) \right \|_{\ell^1} \le C \hhbar \label {F45}
\ee
for any ${\Vector{c}} \in \ell^1_\R$ such that $\| {\Vector{c}} \|_{\ell^1} \le \delta_0$.

\item [iv.] the map $\Vector {r}_4 : \ell^1_\R \times H^1 \to \ell^1_\R$ satisfies
\be 
\left \| \left ( D_{\Vector{c}} \Vector{r}_4 \right ) \left ( \Vector{c}, \left ( \hat \psi_\perp (\Vector {c} ) \right ) \right ) \right \|_{{\mathcal L}
(\ell^1 \to \ell^1)} \le   C \hhbar^{1/{4} } \| \Vector {c} \|_{\ell^1} \, . \label {F46}
\ee
\end {itemize}

In conclusion 
\bee
\|D_{\Vector{c}} \Vector{r} (\Vector{c})\|_{{\mathcal L}(\ell^1 \to \ell^1)}
\le C \left [ \hhbar^{{9}/{4}} \|\Vector{c}\|_{\ell^1} + \hhbar^3 + (1+\|\Vector{c}\|_{\ell^1}^2 ) 
e^{-(S_0-\rho)/\hhbar} \right ] .
\eee
\end {lemma}

\begin {proof}
Estimate (\ref {F43}) has been already proved (see estimate (37) by \cite {FS2}). \ Concerning $\Vector {r}_2$ 
we recall that it is the linear map defined in Lemma \ref {Lemma2}, hence $D_{\Vector {c}} \Vector {r}_2$ 
is independent of $\Vector {c}$ and such that (see Lemma \ref {Lemma8}.iv): $\| D_{\Vector {c}} \Vector {r}_2\|_{\ell^1} = \aasy 
\left ( e^{-S_0/\hhbar } \right )$, \ Similarly for $\Vector {r}_5$ as defined by (\ref {F13}). \ Concerning 
$\Vector {r}_3$ we recall that is is defined in Lemma \ref {Lemma2} and it does not directly depend on $\Vector {c}$, 
furthermore the estimate (\ref {F44}) on the $\ell^1$-norm comes from the fact that $\Vector {r}_3$ linearly depends 
of $\psi_\perp$ and from Lemma \ref {Lemma5}. \ Concerning the term $\Vector {r_4}$ it is defined as 
$r_4^n = \langle \w_n , \psi^3 \rangle - c_n^3 C_1$; then immediately follows that the map 
$\Vector {c} \to \Vector {r}_4 \left ( \Vector {c} , \hat \psi_\perp (\Vector {c} ) \right )$ is smooth. \ Furthermore, 
a straightforward calculation yields to the following expression
\bee
r_4^n &:=& r_4^n (\Vector {c}, \tilde \psi_1 , \psi_\perp ) = \langle u_n , \psi_\perp^3 \rangle + 
\langle u_n , \tilde \psi_1^3 \rangle + 3 \langle u_n , \psi_\perp (c_n u_n +\tilde \psi_1)^2 \rangle + \\ 
&& \ \ +3 \langle u_n , (c_n u_n + \tilde \psi_1 ) \psi_\perp^2 \rangle + 3 \langle u_n ,c_n^2 u_n^2 \tilde \psi_1 \rangle + 
3 \langle u_n , c_n u_n \tilde \psi_1^2 \rangle 
\eee
where we set $\tilde \psi_1 = \psi_1 -c_n u_n = \sum_{m\not= n} c_m u_m$. \ Since $u_n \tilde \psi_1 = \aasy (e^{-S_0/\hhbar })$ 
by Lemma \ref {Lemma8}.iv then the leading term in $r_4^n$ is given by 
\bee
r_4^n (\Vector {c}, 0 , \psi_\perp ) = \langle u_n , \psi_\perp^3 \rangle +  3 c_n^2 \langle u_n , \psi_\perp u_n^2 \rangle + 
3 c_n \langle u_n , u_n \psi_\perp^2 \rangle \, . 
\eee
From this fact and because $\| \psi_\perp \|_{H^1} \le C \hhbar^{1/2}$ (Theorem \ref {Thm5}), $\| D_{\Vector {c}} \hat \psi_\perp 
\|_{{\mathcal L}(\ell^1 \to H^1 )} \le C \hhbar $ (Remark \ref {Rem15}), $\| u_n \|_{L^\infty } \le C \hhbar^{-1/4} $ 
(Lemma \ref {Lemma8}.vi) and Lemma \ref {Lemma8}.v then it follows that the leading term in $D_{\Vector {c}} {\Vector {r}_4}$ 
is estimated by 
\bee
6 \| \Vector {c} \|_{\ell^1} \max_n \| u_n \|_{L^2}^2 \| \psi_\perp u_n \|_{L^\infty} \le C \hhbar^{1/4} \| \Vector {c} \|_{\ell^1} \, . 
\eee
By collecting all these facts and since (\ref {F3}) the the proof follows.
\end {proof}

Now, we fix $\delta_0 \ge 1$, then 
\be
\sup_{\| \Vector{c} \|_{\ell^1} \le \delta_0 } \| \Vector{r} \|_{\ell^1} \le C \hhbar^{9/4} \ \mbox { and } \  
\sup_{\| \Vector{c} \|_{\ell^1} \le \delta_0 } \| D_{\Vector {c}} \Vector{r} \|_{\ell^1} \le C \hhbar^{9/4} \, . \label {F47}
\ee

\begin {lemma} \label {Lemma7}
Let $\Vector{g}^S$ and $\tilde \lambda$ be a solution to equation ${\mathcal G} (\Vector{g}^S ,0)=0$, as given by 
Theorem \ref {Thm4}; the linear map $D_{\Vector {c}} {\mathcal G} (\Vector{g}^S ,0)$ is one-to-one and onto.
\end {lemma}

\begin {proof}
Again, let us omit the upper letter $S$ when this does not cause misunderstanding. \ By construction, the linear map 
\bee
D_{\Vector {c}} {\mathcal G} (\Vector{g}^S ,0): \ell^1_{\R} \to \ell^1_{\R} 
\eee
is associated to a tridiagonal matrix defined as 
\bee
\mbox {tridiag} \left ( \beta , \tilde \lambda  - f   \Chi (n) - 3 \nu g_n^{2 } , \beta \right ) 
\eee
Here, we make use of the result given in Appendix A by \cite {ABK}; in particular, because $\beta$ is exponentially small as $\hhbar$ goes to zero we only have 
to check that 
\be
\left | \tilde \lambda - f   \Chi (n) - 3 g_n^{2} \nu \right | \ge C \hhbar^{3/2} >0 \label {F48}
\ee
uniformly holds true with respect to $n$, where $\Vector {g}^S$ is close to $\Vector {d}^S$ and $\tilde \lambda$ is close to 
$\mu^S$. \ Indeed, the left hand side of (\ref {F48}) turns out to be close to $|\tilde \lambda T_n |$, where $T_n$ 
is given by (\ref {F34}) and, by Lemma \ref {Lemma4}, it is such that $|T_n| >\frac 12 $ for any $n$; furthermore $\tilde \lambda 
\sim \hhbar^{3/2}$. \ From this fact and from the argument given in Appendix A by \cite {ABK} then the linear map $D_{\Vector {c}} 
{\mathcal G} (\Vector{g}^S ,0)$ is one-to-one and onto.
\end {proof}

Now, we are ready to conclude the proof of Theorem \ref {Thm6}. \ Le $\Vector {g}^S$ be the solution to (\ref {F14}) associated 
to the finite-mode solution $\Vector {d}^S$ satisfying the assumptions of Theorem \ref {Thm6}. \ By the Implicit Function Theorem, there exist an $\hhbar$-independent $\delta >0$ such that if $|\kappa | \le \delta $ then there 
exists a  unique solution $\Vector {c}(\kappa )$ in a $\ell^1$-neighborhood of $\Vector {g}^S$ satisfying $\mathcal{G} (\Vector{c}, \kappa )=0$. \ Then we can 
conclude that there exists $\hhbar^\star >0$ such that for any $\hhbar < \hhbar^\star$ 
there is a unique solution $ \Vector {c}^S \in \ell^1_\R$ to $\mathcal{G} (\Vector{c}, \hhbar^2 )=0$, and it is such that
\be
\| \Vector {c}^S - \Vector {g}^S \|_{\ell^1} = \asy (\hhbar^{3/4} ) \label {F49}
\ee
since $|T_n| > C \hhbar^{3/2}$ and (\ref {F47}). \ Then the Theorem follows where $\psi^S = \psi_1^S + 
\psi_\perp^S$, $\psi_1^S = \sum_{n\in \Z} c_n^S u_n$ and $\psi_\perp^S = \hat \psi_\perp (\Vector{c}^S )$, furthermore 
\bee
\left \| \psi^S - \sum_n d_n^S \w_n \right \|_{H^1} &\le & \| \psi_\perp^S \|_{H^1} + \left \| \psi_1^S - \sum_n d_n^S \w_n \right \|_{H^1} \\ 
&\le &  \| \psi_\perp^S \|_{H^1} + \left \|  \sum_n (c_n^S -d_n^S) \w_n \right \|_{H^1} \\ 
&\le &  \| \psi_\perp^S \|_{H^1} + \left [ \left \| \Vector {c}^S - \Vector {g}^S \right \|_{\ell^1} + \left \| \Vector {g}^S - 
\Vector {d}^S \right \|_{\ell^1} \right ] \| \w_0 \|_{H^1} \\ 
& \le & C \hhbar^{1/2} + C \left [ \hhbar^{3/4} + \aasy \left ( e^{-S_0/\hhbar } \right ) \right ] \hhbar^{-1/2} \le  C \hhbar^{1/4} 
\eee
because of Theorems \ref {Thm4}, \ref {Thm5}, equation (\ref {F49}) and Lemma \ref {Lemma8}.vi. \ Theorem \ref {Thm6} is so proved.   
\end {proof}

\appendix

\section {Proof of Lemma \ref {Lemma3}}

Let us recall that $\psi =\psi_1 + \psi_\perp $ satisfies to the following estimates (see (28) and (29) in \cite {FS2})
\bee
\| \nabla \psi_1 \|_{L^2} \le C \hhbar^{-1/2} \| \Vector {c} \|_{\ell^1}
\eee
and
\bee
\| \psi_1 \|_{L^\infty} \le C \hhbar^{-1/4} \| \Vector {c} \|_{\ell^1} \, .
\eee
Moreover, from Lemma 1, the following inequalities hold true:
\bee
\| u_m u_n \|_{L^1} \le C e^{-[(S_0-\rho')|m-n|-\rho'']/\hhbar} \, , \ m\not= n 
\eee
\bee
\left \| \sum_n |u_n| \right \|_{L^\infty} \le C \hhbar^{-1/2}
\eee
\bee
\| u_n \|_{L^p} \le C \hhbar^{-(p-2)/4p} \ \forall p \in [2,\infty] \ \mbox { and } \ \| \nabla u_n \|_{L^2} \le C \hhbar^{-1/2} \, . 
\eee
Furthermore, we recall also the following Sobolev inequality (see Theorem 8.8 in Brezis)
\bee
\| u \|_{L^\infty (\R )} \le C \| u \|_{H^1 (\R )} \, , \ \forall u \in H^1 (\R ) \, .
\eee
Now, let 
\bee
r_4^n = \langle u_n , \psi^3 \rangle - C_1 c_n^3 = f_1^n + f_2^n 
\eee
where we set
\bee
f_1^n = \langle u_n , \psi_1^3 \rangle - C_1 c_n^3
\eee
and
\bee
f_2^n = \langle u_n , \psi^3 - \psi_1^3 \rangle \, . 
\eee
By the proof of Lemma 3 in \cite {FS2} we have that for any $\rho \in (0,S_0)$ there exists $C:=C_\rho >0$ such that the vector $\Vector {f}_1 = \{ f_1^n \}_{n\in \Z}$ 
can be estimated as follows 
\bee
\| \Vector {f}_1 \|_{\ell^1} = \sum_n |f_1^n| \le C \| \Vector {c} \|^3_{\ell^1} e^{-(S_0-\rho )/\hhbar }\, .
\eee
For what concerns the term $\Vector {f}_2 = \{ f_2^n \}_{n\in \Z}$ we observe that
\bee
f_2^n &=& \langle u_n , (\psi_1 + \psi_\perp )^3 - \psi_1^3 \rangle = \langle u_n , \psi_\perp^3 \rangle + 3 \langle u_n , \psi_\perp^2 \psi_1 \rangle + 
3 \langle u_n , \psi_\perp \psi_1^2 \rangle
\eee
where
\bee
\sum_n \left | \langle u_n , \psi_\perp^3 \rangle \right | &\le &  \left \langle \sum_n| u_n |, |\psi_\perp|^3 \right \rangle \\
&\le & \left \| \sum_n| u_n | \right \|_{L^\infty} \| \psi_\perp \|_{L^2}^2  \| \psi_\perp \|_{L^\infty} \le C \hhbar^{-1/2} \| \psi_\perp \|_{H_1}^3
\eee
\bee
\sum_n \left | \langle u_n , \psi_\perp^2 \psi_1 \rangle \right | &\le & \sum_{n,m} |c_m| \, \left | \langle u_n, u_m \psi_\perp^2 \rangle \right | \\
&\le & \sum_{n,m} |c_m| \, \| u_n u_m \|_{L^1} \| \psi_\perp \|_{L^\infty}^2 \\ 
&\le & \sum_{n,m} |c_m| \, \| u_n u_m \|_{L^1} \| \psi_\perp \|_{H^1}^2 \\
&\le & \sum_{m} |c_m|  \left [ \sum_n \| u_n u_m \|_{L^1} \right ] \| \psi_\perp \|_{H^1}^2 \\
&\le & C \| \Vector {c} \|_{\ell^1} \| \psi_\perp \|_{H^1}^2
\eee
\bee
\sum_n \left | \langle u_n , \psi_\perp \psi_1^2 \rangle \right | 
&\le & \sum_{n,m,\ell } |c_m| \, |c_\ell | \, \left | \langle u_n u_m u_\ell , \psi_\perp \rangle \right | \\ 
&\le & \sum_n |c_n|^2 \langle |u_n|^3 , |\psi_\perp | \rangle + 2 \sum_n |c_n| \sum_{\ell \not= n} |c_\ell | \langle |u_\ell|^2 |u_n | , |\psi_\perp |\rangle + \\ 
&& \ \ + \sum_n \sum_{m,\ell \not= n} |c_m| \, |c_\ell | \langle |u_m|\, |u_n|\, |u_\ell |, |\psi_\perp |\rangle \\ 
&\le & \max_n \left [ \| u_n \|_{L^2}^2 \| u_n \|_{L^\infty } \right ] \sum_n |c_n|^2 \| \psi_\perp \|_{L^\infty} + \\ 
&& \ \ + 2 \sum_n |c_n| \sum_{\ell \not= n } |c_\ell | \, \| u_n u_\ell \|_{L^1} \| u_\ell \|_{L^\infty} \| \psi_\perp \|_{L^\infty} + \\ 
&& \ \ + \sum_n \sum_{m,\ell \not= n} |c_m|\, |c_\ell | \| u_n u_m \|_{L^1} \| u_\ell \|_{L^\infty} \| \psi_\perp \|_{L^\infty} \\ 
&\le & \| \Vector {c} \|_{\ell^1}^2 \hhbar^{-1/4} \| \psi_\perp \|_{H^1} + 2 \| \Vector {c} \|_{\ell^1}^2 e^{-(S_0 - \rho )/\hhbar} \hhbar^{-1/4} \| \psi_\perp \|_{H^1} + \\
&& \ \ + \| \Vector {c} \|_{\ell^1}^2 \hhbar^{-1/4} e^{-(S_0 - \rho )/\hhbar}  \| \psi_\perp \|_{H^1} \\ 
&\le & C \| \Vector {c} \|_{\ell^1}^2 \hhbar^{-1/4}  \| \psi_\perp \|_{H^1}
\eee
Therefore
\bee
\| \Vector {f_2} \|_{\ell^1} \le C \left [ \hhbar^{-1/2} \| \psi_\perp \|_{H_1}^3 + \| \Vector {c} \|_{\ell^1} \| \psi_\perp \|_{H_1}^2 +  \| \Vector {c} \|_{\ell^1}^2 
\hhbar^{-1/4}  \| \psi_\perp \|_{H^1} \right ] \, . 
\eee
Hence,
\bee
\| \Vector {r_4} \|_{\ell^1} &\le & \| \Vector {f_1} \|_{\ell^1} + \| \Vector {f_2} \|_{\ell^1} \\ 
&\le & C \left [ \hhbar^{-1/2} \| \psi_\perp \|_{H_1}^3 + \| \Vector {c} \|_{\ell^1} \| \psi_\perp \|_{H_1}^2 +  \| \Vector {c} \|_{\ell^1}^2 
\hhbar^{-1/4}  \| \psi_\perp \|_{H^1} + \| \Vector {c} \|^3_{\ell^1} e^{-(S_0-\rho )/\hhbar } \right ] \, . 
\eee

\end {document}